\documentclass[accepted]{uai2023} 

\usepackage[american]{babel}

\usepackage{natbib} 
\usepackage{mathtools} 
\usepackage{booktabs} 
\usepackage{tikz} 

\usepackage{amsthm}
 \usepackage{algorithm2e}

\newcommand{\red}[1]{\textcolor[rgb]{1,0.1,0.1}{#1}}

\newtheorem{definition}{Definition}

\newtheorem{corollary}{Corollary}
\newtheorem{theorem}{Theorem}
\newtheorem{assumption}{Assumption}
 \newcommand{\indep}{\perp\!\!\!\!\perp} 
\newcommand{\oo}{{\circ\!{\--}\!\circ}}

\title{Causal discovery for time series from multiple datasets with latent contexts}

\author[1]{\href{mailto:<wiebke.guenther@dlr.de>?Subject=UAI 2023 CD from multiple datasets with latent contexts}{Wiebke~G\"{u}nther}{}}
\author[2,1]{Urmi~Ninad}
\author[1,2]{Jakob~Runge}
\affil[1]{%
    German Aerospace Center\\
    Institute of Data Science\\
    07745 Jena, Germany
}
\affil[2]{%
    Technische Universit\"{a}t Berlin\\
    Dept. of Electrical Engineering and Computer Science\\
    10623 Berlin, Germany
}

\begin{document}
\maketitle

\begin{abstract}
Causal discovery from time series data is a typical problem setting across the sciences. Often, multiple datasets of the same system variables are available, for instance, time series of river runoff from different catchments. The local catchment systems then share certain causal parents, such as time-dependent large-scale weather over all catchments, but differ in other catchment-specific drivers, such as the altitude of the catchment. These drivers can be called temporal and spatial contexts, respectively, and are often partially unobserved. Pooling the datasets and considering the joint causal graph among system, context, and certain auxiliary variables enables us to overcome such latent confounding of system variables. In this work, we present a non-parametric time series causal discovery method, \emph{J(oint)-PCMCI$^+$}, that efficiently learns such joint causal time series graphs when both observed and latent contexts are present, including time lags. We present asymptotic consistency results and numerical experiments demonstrating the utility and limitations of the method.
\end{abstract}

\section{Introduction}\label{sec:intro}
Causal discovery from observational data has gained widespread interest in recent years. Next to score-based methods~\cite{chickering2002learning}, Granger causality \cite{granger1969investigating}, and the more recent restricted structural causal models (SCM) framework \cite{peters2017elements,spirtes2016causal}, the constraint-based approach\citep{spirtes2000causation} to this discovery task exploits conditional independencies in the data to constrain causal graphs and can flexibly handle nonlinear dependencies.

Most real world data comes in the form of \emph{time series}, which provide opportunities and challenges for causal discovery~\citep{runge2019inferring}. While the inherent time order implies certain causal directions, time-series data typically violates the \emph{i.i.d.} assumption usually made in  conditional independence testing. Therefore, specific algorithms that target the challenges of time-series data have become an increasingly popular sub-category within causal discovery, for instance, versions of the PC algorithm and FCI \citep{entner2010causal,malinsky2018causal}, or the PCMCI framework~\citep{runge2019detecting,runge2020discovering,gerhardus2020high}. 

The methods mentioned above consider single multivariate time series datasets and aggregate samples across time. Another relevant development has been the incorporation of multiple datasets and modeling their different contexts~\citep{mooij2020joint,huang2020causal}, which can also be framed as a data-fusion problem~\citep{pearl2011transportability,bareinboim2016causal}. 

In the following, we illustrate the main ideas of this paper on the example of time series datasets of river runoff from different catchment systems~\citep{wagener2007catchment}. If we assume all of these to come from the same (stationary) distribution, we can just concatenate (pool) the data to obtain a larger sample and, hence, more reliable causal discoveries among the system variables. But multiple datasets can also be used to de-confound relationships: The local catchment systems often share certain causal parents, such as time-dependent large-scale weather dynamics over all catchments, that can be called temporal contexts and cause latent confounding between two or more system variables, if they are unobserved. If we now assume that such a time-dependent latent confounder is the same across all datasets, we can condition on the time index (or add a so-called \emph{time-dummy variable}) by aggregating samples across datasets, instead of aggregating across time. This then yields a joint graph across all time points.

But datasets not only share common causal drivers, they also differ in other dataset-specific drivers, such as the altitude or vegetation-type of the catchment, which can be called spatial contexts. As spatial contexts are constant across time, they do not constitute a confounding \emph{within} each dataset. However, in the pooled data they vary across datasets and lead to confounding in the joint graph across datasets. If spatial contexts are observed, they can be included as variables in the analysis and if they are unobserved, they can be de-confounded by the same idea as time-dependent confounders by assuming that they are constant across all time points and conditioning on the dataset index (or adding a so-called \emph{space-dummy variable}, not to be confused with physical space). The idea of conditioning on time and space is heavily employed in \emph{fixed-effect panel regression models} in econometrics~\citep{angrist2009mostly} and here we consider these for causal discovery. 

Next to deconfounding system variables, observed context variables can help to orient causal links: Consider two system variables $X\oo Y$ whose causal direction cannot be identified by Markov equivalence. In the joint graph we could add a context variable $C$ and assume (or learn) that $C\rightarrow X \oo Y$. Then applying the collider or orientation rules~\citep{meek1995causal} allows to infer the causal direction between $X$ and $Y$.

Our approach partially follows \citet{mooij2020joint}. Their approach is to pool data from different contexts, for instance, observational and interventional data, and do causal discovery on the pooled dataset, called \emph{joint causal inference} (JCI). In particular, they established a general framework to (i)~interpret contexts as auxiliary variables that describe the context of each dataset, (ii)~pool all the data from different contexts while keeping the contextual information of the data by including the auxiliary context variables into a single dataset, and (iii)~apply standard causal discovery to all data jointly, incorporating appropriate background knowledge on the causal relationships involving the context variables.

Our aim is to extend the PCMCI$^+$ time series causal discovery algorithm~\citep{runge2020discovering} to the case of datasets from multiple dataset- or time dependent contexts with potentially unobserved context confounders of the system variables. We term this technique \emph{J(oint)-PCMCI$^+$} because it combines the two JCI-ideas mentioned above, i.e., pooling datasets from multiple contexts, and adding observed context variables to the graph. We go beyond JCI by providing a specific algorithmic implementation of it in the time-series setting and using time- and space-dummy variables to account for latent context variables that confound system variables, that are conceptually similar to surrogate variables in \citep{huang2020causal}. Therefore, we are faced with the additional challenge of dealing with observed context as well as dummy variables simultaneously, which requires caution due to the fact that they are deterministically related to one another \cite{lemeire2012conservative}. This approach combines the advantages of PCMCI$^+$ regarding detection power and false positive control in the presence of strong autocorrelation~\citep{runge2020discovering} with the advantage of the JCI framework.
To summarize, we present a consistent causal discovery algorithm J-PCMCI$^+$ that can:
\begin{enumerate}
    \item[(i)] de-confound those system nodes that are confounded by latent contexts without having any knowledge of the latent contexts themselves;
    \item[(ii)] retain as much information about the causal links between the observed context and system variables as possible by checking conditional independencies appropriately;
    \item[(iii)] discover the correct induced causal graph between the system nodes.
\end{enumerate}

\section{Related Work}\label{sec:related}

In causal inference, the idea of context variables has been explored under different notions :  “policy variables” \cite{spirtes2000causation}, “force variables” \cite{pearl1993comment}, “decision variables” in influence diagrams \cite{dawid2002influence}, “selection variables” in selection diagrams \cite{bareinboim2013general}, and “environment variables” \cite{peters2016causal}. \cite{mooij2020joint} established a general framework to combine data from multiple contexts with traditional causal discovery techniques. 

In \cite{huang2020causal}, the problem of \emph{heterogeneous} data, which might correspond to varying dataset collection conditions (analogous to contexts in JCI), as well as the problem of \emph{non-stationary} data was addressed in a framework called CD-NOD.
Here, changing causal mechanisms across time or datasets were interpreted as confounding of the system by an unobserved \emph{pseudo-confounder}, so named because it can be written as a function of the dataset or the time index. This confounding was then addressed by introducing a \emph{surrogate variable} that captures changes of causal mechanisms, thereby deconfounding the pseudo-confounded system variables. Further, they employed the information of changing causal mechanisms to infer additional causal directions than standard causal discovery allows for by formalizing independence of cause and mechanism \cite{peters2017elements}. Note that explicitly known context variables can not be included in this setup. It also focuses on non-stationarity that can be modeled as a smooth function of time. This is a restriction we do not place.

The concept of changing regimes over time is similar to the presence of a temporal context. In particular, \citet{saggioro2020reconstructing} consider regimes that vary over time and assume that these regimes are not known a priori. This is in contrast to our assumptions that the context variables are constant over time or across the data sets which implies that we know from domain knowledge when a context change might happen. They present the  Regime-PCMCI algorithm that learns the regimes together with the causal graphs within each regime.

Certain Bayesian methods can also deal with heterogeneous data, i.e.\ data from different contexts, e.g. \cite{zhou2022causal}. Naturally, identifying the causal structure by a Bayesian method requires strong model assumptions. Instead of using dummy or surrogate variables, the authors suggest to impute possible latent covariates using an embedding method, and they also provide a way to infer the latent covariates jointly with the causal graph.

Constraint-based causal discovery methods for time-series data historically began with Granger causality \cite{granger1969investigating}, and since has been addressed in \citet{enter2010causal, malinsky2018causal, runge2019detecting, runge2020discovering, gerhardus2020high} to cover non-linear relationships, contemporaneous as well as lagged links, latent confounders and highly auto-correlated data. For an overview, see \cite{runge2019inferring}.

\section{Theoretical Foundations} \label{sec:found}
\begin{figure}[t]
    \centering
    \includegraphics[width=0.75\linewidth]{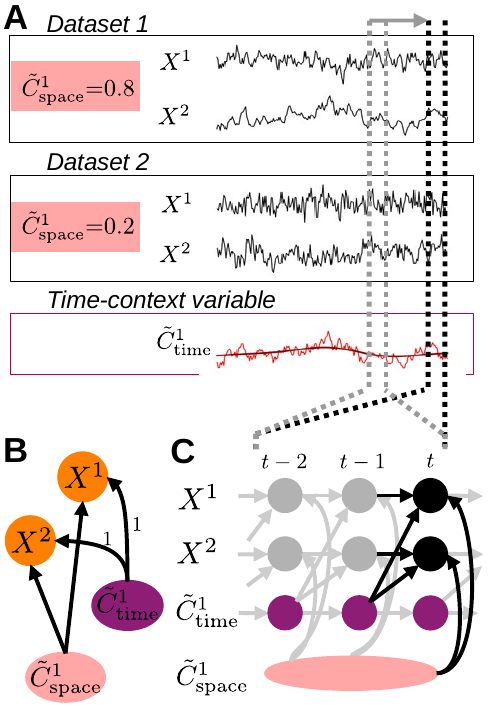}
    \caption{\textbf{Causal discovery with temporal- and spatial-contexts.} \textbf{(A)} Two datasets of system variables $X^1,X^2$ may be confounded by the same temporal context $\tilde{C}^1_{\rm time}$, but differ in dataset-specific characteristics that are constant over time, here autocorrelation, which can be represented by a spatial context $\tilde{C}^1_{\rm space}$. J(oint)-PCMCI$^+$ learns the causal relations which can be represented in \textbf{(B)} a summary causal graph (link labels denote time lags) or \textbf{(C)} a time series graph. Samples of the nodes are pooled from both datasets over the (user-defined) stationary part of the time series (grey dotted lines) leading to a repeating structure of the time series graph (grey links). Context nodes can also help orienting links since they can create colliders. If contexts are unobserved, J-PCMCI$^+$ can utilize temporal- or spatial dummy variables to de-confound system variable relationships.}
    \label{fig:intro}
\end{figure}

Within the JCI framework~\citep{mooij2020joint}, the causal relations of a system and its context are represented by a joint (or meta) structural causal model (SCM).
A \emph{system variable} is a time-dependent random variable whose distribution can change across datasets. In the following, a \emph{temporal context variable} is a time-dependent random variable~\footnote{The term \emph{random variable} is used for context variables in the sense that it is done in \cite{mooij2020joint}.} that remains the same across datasets. A \emph{spatial context variable} is a random variable that is constant over time and within a dataset but can change across datasets. See figure \ref{fig:intro} for an example. The information on which variables belong to the system or to the context is given as a domain assumption. We now formulate an assumption on the data-generating mechanism that always holds, unless stated otherwise.

\begin{assumption}[Joint time-dependent SCM]\label{ass:SCM}
    The underlying data-generating mechanism across datasets $d\in \mathcal{D}$ with $|\mathcal{D}|=M$ is an acyclic time-dependent structural causal model (SCM) involving the time-dependent system variables $\mathbf{X}_t=\{X_t^i\}_{i\in \mathcal{I}}$ at time $t$ as well as context variables $\tilde{\mathbf{C}}=\tilde{\mathbf{C}}_{time}\dot{\cup} \ 
    \tilde{\mathbf{C}}_{space}$ with temporal contexts $\tilde{\mathbf{C}}_{time,t}=\{\tilde{C}_t^k\}_{k\in \mathcal{K}_{\rm time}}$ and (time-independent) spatial-contexts $\tilde{\mathbf{C}}_{space}=\{\tilde{C}^l\}_{l\in \mathcal{K}_{\rm space}}$, for $i\in \mathcal{I}$:
    
        \begin{equation} \label{eq:scm}
            \begin{split}
            &\mathbf{X}^{d}_t:= \mathbf{f}({Pa}_X(\mathbf{X}^{d}_t),\operatorname{Pa}_{\tilde{C}_{\rm time}}(\mathbf{X}^{d}_t), \operatorname{Pa}_{\tilde{C}_{\rm space}}(\mathbf{X}^{d}_t), \boldsymbol{\eta}^{d}_t) \\ 
            &\mathbf{\tilde{C}}_{{\rm time}, t} := \mathbf{g}({Pa}_{\tilde{C}_{\rm time}}(\mathbf{\tilde{C}}_{{\rm time}, t}), \boldsymbol{\eta}_{{\rm time}, t}) \\
            &\mathbf{\tilde{C}}^{d}_{\rm space} := \mathbf{h}({Pa}{\tilde{C}_{\rm space}}(\mathbf{\tilde{C}}_{\rm space}), \boldsymbol{\eta}_{\rm space}^{d})
            \end{split}
        \end{equation}
    where the exogenous noise variables $(\boldsymbol{\eta}^{d}_t, \boldsymbol{\eta}_{{\rm time}, t}, \boldsymbol{\eta}_{\rm space}^{d})$ are jointly independent and $\eta^{i,d}_{t}$ are identically distributed across time and space, $\eta^k_{{\rm time},t}$ are identically distributed across time, and $\eta^{l,d}_{\rm space}$ are identically distributed across space. $\operatorname{Pa}_X$ denotes the causal parents within $\mathbf{X}$, and analogously for $\operatorname{Pa}_{\tilde{C}_{\rm time}}$ and $\operatorname{Pa}_{\tilde{C}_{\rm space}}$.
\end{assumption}

Note that we restrict our exploration to the case of acyclic SCMs. However, since the PC algorithm has been proven to be consistent in the presence of cycles \citep{mooij2020constraint}, we expect that the consistency of our method extends straightforwardly to the case of SCMs with contemporaneous cycles.
As discussed above, we allow for latent context variables. Non-stationary data could be modeled by intervening on a temporal context node simultaneously across all time points.

\begin{assumption}[Markov, Faithfulness and partial causal sufficiency]\label{ass:sufficiency}
    Let $M$ be a SCM of the form (\ref{eq:scm}) with graph $\mathcal{G}$, the joint distribution $P_M(X, \tilde{C})$ induced by the SCM satisfies the Markov Property with respect to the graph $\mathcal{G}$.
    Additionally, $P_M(X, \tilde{C})$ is faithful to the graph $\mathcal{G}$ of $M$.
    The collection of datasets from the joint SCM~\eqref{eq:scm} is assumed to contain data from all system nodes $\mathbf{X}_t$, but may contain unobserved context nodes $\mathbf{L}$ and observed context nodes $\mathbf{C}$,
    \begin{equation}
        \tilde{\mathbf{C}}=\mathbf{L} ~\dot{\cup}~ \mathbf{C}.
    \end{equation}
    Subsequently, $\mathbf{L}$ and $\mathbf{C}$ can also be written as a disjoint union over their spatial and temporal components.
\end{assumption}

Note, that in the setup of \eqref{eq:scm} we have included the assumption that context variables are exogenous to the system, see \emph{JCI assumption 1} in \cite{mooij2020joint}.
Furthermore, we assume that no latent context confounds an observed context and a system variable. Since a latent context is the only possible latent confounder in our setup, this assumption is equivalent to \emph{JCI assumption 2} \cite{mooij2020joint}. We make an additional assumption in view of efficiency:

\begin{assumption}[Context-system links]\label{ass:JCI}
    JCI assumption 1 holds.
    Further, a causal link between an observed context and a system variable is not mediated by latent context variables.
\end{assumption}

\begin{assumption}[No context-system confounders]\label{ass:no_conf}
    JCI assumption 2 holds, i.e.\ no latent context confounds an observed context and a system variable.
\end{assumption}

Finally, we adapt the \emph{pseudo-causal sufficiency assumption} of \cite{huang2020causal} for our case as follows:
\begin{assumption}[Context-determinism]\label{ass:asymmetry}
    The spatial context variables are deterministic functions of the dataset index.
    The temporal context variables are deterministic functions of the time-index.
    We assume these functions to be non-invertible.
\end{assumption}

The SCM \eqref{eq:scm} yields a joint time series causal graph $\mathcal{G}$ (see Fig.~\ref{fig:intro}) over all datasets $d$. Note that the spatial context variables appear as a single node in the time series graph. This representation is chosen to denote that they are constant in time and therefore do not have a time-dimension associated with them.
The joint graph $\mathcal{G}$ is related to the target of our discovery task:
\begin{definition}[Target graph]\label{def:target_graph}
   The target graph of J-PCMCI+ is the induced subgraph of $\mathcal{G}$ over the system nodes together with the observed context nodes and their edges to the system nodes.
\end{definition}

\section{Method}\label{sec:method}
The general idea of our method is to include context nodes in the time series graph motivated by the JCI approach. In order to deal with latent context variables, we introduce dummy variables (Sect.~\ref{subsec:def}) before presenting J-PCMCI$^+$ (Sect.~\ref{subsec:alg}) and state consistency results in Sect.~\ref{subsec:theorems}.

\subsection{Dummy variables as proxies for latent confounders}\label{subsec:def}
Let $\mathcal{G}=(V, E)$ be the time series graph corresponding to SCM \eqref{eq:scm}, where $V$ denote the vertices and $E$ edges between vertices. 
Here, $V = \mathbf{X} \cup \mathbf{C} \cup \mathbf{L}$, where $\mathbf{X}$ (resp.~$\mathbf{C}$ and $\mathbf{L}$) refers to vertices at all time points.
The set of edges $E$ can be written as a disjoint union between edges $E_L$, where at least one of the corresponding nodes is in $\mathbf{L}$, and its complement $E_O$, which consists of edges that only connect observed (system and/or context) nodes. That is, $E = E_L \dot{\cup} E_O.$ Further, $E_O$ is itself a disjoint union of $E_{S}$, which are the edges where at least one of the two nodes is in $\mathbf{X}$ and its complement $E_C$, i.e., $E_O = E_S \dot{\cup} E_C$. 
Using this notation, we can express the definition of the target graph (definition \ref{def:target_graph}) more formally: Based on a given ground truth graph $\mathcal{G} = (\mathbf{X} \dot{\cup} \mathbf{C} \dot{\cup} \mathbf{L}, E_L \dot{\cup} E_S \dot{\cup} E_C)$, we define the target graph $\Tilde{\mathcal{G}}$ as $\Tilde{\mathcal{G}} = (\mathbf{X} \dot{\cup} \mathbf{C}, E_S)$. See figure 1 in the SM for an example.

\begin{definition}[Space dummy variable]\label{def:space_dummy}
    The space dummy variable $D_{\rm space} $, henceforth referred to as space dummy, is a variable that labels datasets.
\end{definition}

Without prior expert knowledge, this labelling is arbitrary. 
For instance,  the simplest embedding is $D_{\rm space} \in \{1, \ldots, M\}$.
Alternatively, in a one-hot-encoded embedding, $i \in \{1, \ldots, M\}$ denotes the position of the $1$ in an $M$-dimensional vector where all other entries are $0$.
In the following, we work with a one-hot-encoded space dummy. However, we note that the question of which embedding to choose for the space dummy is far from settled, and requires further expert knowledge about the particular setup and the type of conditional independence to be used in the causal discovery algorithm. Refer to the SM for further details.

Also note, there is no one-to-one relation between datasets and spatial contexts, i.e., two datasets can have same value for a spatial context.

\begin{definition}[Time dummy variable]\label{def:time_dummy}
    The time dummy variable $D_{\rm time}$, henceforth referred to as time dummy, is a variable that labels each time-step in the time-series data. 
\end{definition}

Here too, we arbitrarily choose the embedding for $D_{\rm time} $ to be a one-hot-encoding into a $T$-dimensional vector, where $T$ is the length of the time-series.

\begin{definition}[Dummy projection]\label{def:dummy_proj}
    We define the dummy-projection of the graph $\mathcal{G}=(\mathbf{X} \cup \mathbf{C} \cup \mathbf{L}, E_L \cup E_O)$ to be the graph $\mathcal{G}_D = (\mathbf{X} \cup \mathbf{C} \cup \{D_{\rm space}  \cup D_{\rm time} \}, \tilde{E})$, where edges $\tilde{E}$ are defined as:
    \begin{equation*}
        \begin{split}
            \tilde{E} &= \{ (D_{\rm space} ,v) | (u,v) \in E, \ \forall \  u \in \mathbf{L}_{space} \text{ and } v \in \mathbf{X} \} \\
            \cup & \{ (D_{\rm time} ,v) | (u,v) \in E, \ \forall \ u \in \mathbf{L}_{time} \text{ and } v \in \mathbf{X}_t \}\cup E_S.
        \end{split}
    \end{equation*}
\end{definition}

Further, note that in the dummy projection, we have omitted the edges $E_C$, i.e., the edges between the observed context variables, since these relationships are not of interest for the target graph (definition \ref{def:target_graph}).

Finally, we introduce the dummy-deleted graph. For a visualization of the dummy projection and deletion operations, see figure 3 in SM. 

\begin{definition}[Dummy deletion]\label{def:dummy_del}
    Let $\mathcal{G}_D$ be the dummy projection of the graph $\mathcal{G}$. The dummy-deleted graph $\mathcal{G}_{D_{del}}$ is the graph where the dummy variables and any outgoing edge therefrom is removed.     
\end{definition}

Under assumption \ref{ass:asymmetry}, for $C_{s} \in \mathbf{C}_{space}$ and $L_s \in \mathbf{L}_{space}$, we can always find (not necessarily unique) functions $g_C$ and $g_L$ with $C_s = g_C(D_{\rm space} )$ and $L_s = g_L(D_{\rm space} )$, and analogously for the temporal counterparts.

Such mappings can be assumed to exist, since the dummy $D$ takes a unique value within each dataset (def. \ref{def:space_dummy}), and each spatial context variable $C$ is assumed constant within each dataset (and analogously for the temporal version). Therefore, there exists a mapping $g$ with $C=g(D)$.
There would exist no mapping from the dummy variable to the context variable if the context variable would take two different values within one dataset, but this case is excluded by assumption.

The introduction of the dummy variables into the causal discovery task is for the purpose of removing the confounding effect of latent context variables on a pair of system variables. Latent system variables and their confounding cannot be handled with this method. As we will see in Section \ref{subsec:alg}, our method yields a graph between $\mathbf{X}$, $\mathbf{C}$ and $D$, which is not exactly the dummy projection of the true graph $\mathcal{G}$, but whose dummy deletion is the target graph (definition \ref{def:target_graph}).

\paragraph{Interpretation of the dummy-projection}
Some caution has to be applied in the interpretation of links between the dummy and system nodes in the dummy-projected graph. These links are only placeholders for the links between unobserved context nodes and system nodes in the ground truth graph. The dummy is not a causal variable itself! 
Further, note that in definition \ref{def:dummy_proj}, we did not include links between the dummy and observed context nodes in the dummy projection since these links are deterministic by assumption~\ref{ass:asymmetry} and always present and thus not informative.

\paragraph{Including both observed context and dummy variables}
A natural question at this point might be, why include the observed context variables at all in the causal discovery task, even though the dummy can remove all influence of context variables because of the general way in which is it defined (definitions \ref{def:space_dummy}, \ref{def:time_dummy}). Theoretically, there would indeed be nothing wrong with excluding contexts altogether. However, the dummy variable is not interpretable, and not useful when the goal is to learn the causal relationship between specific context and system variables.
Further, as we will see below, in the first step our causaldiscovery algorithm learns the context-system adjacencies, the relationship between which may be mathematically simpler than those between the highly general dummy and system variables. In the second step, it learns the dummy-system adjacencies given the context parents learnt in the first step. This helps infer influence on the system variable of the dataset label that cannot be explained by the context variables, i.e., in essence we learn hidden contexts. Finally, as also pointed out by \citet{mooij2020joint}, the separate observed context nodes help in orienting adjacencies between system variables. Refer also figure \ref{fig:intro} for a visualization.
 
A naive implementation to learn the causal relationships between context-system and dummy-system, where the context and dummy variables are treated on the same footing when testing adjacencies to the system, would not be correct as the \emph{causal faithfulness assumption} would be violated. This is because the relationship between context and dummy variables is deterministic.
The two-step procedure outlined above circumvents erroneous inferences of adjacencies due to faithfulness violation. For details, see sect.\ref{subsec:alg} and SM.

\subsection{Algorithm}\label{subsec:alg}
In the pooled dataset, we include one variable for the space and time dummy each, as well as for spatial context variables, at time $t$. These nodes can only have contemporaneous links to system nodes since they either do not change over time or contain no information about the temporal structure.
That is, only temporal context variable can have a lagged influence on the system variables, see also figure \ref{fig:intro} and SM (figure 2). 
To be able to deal with observed contexts and dummy variables, that are essentially placeholders for the unobserved context variables, our method first discovers links between system and observed context nodes while ignoring the dummy nodes, and in the next step discovers links between dummy and system nodes. Finally, using the information on the contextual parents of each system node, we do causal discovery on the system node pairs. In the following, we detail this procedure for the non-time series and time-series case. 

\paragraph{Non-time series case:}
To ease the explanations for the time-series case, where both spatial and temporal context variables can occur, we first focus on the non-time-series case, where only spatial context variables can occur. Consequently, we only have to consider the space dummy. Note that assumption \ref{ass:SCM} can be simplified to the non-time series straightforwardly. We will combine the well-known PC algorithm \cite{spirtes2000causation} with both observed and dummy context variables. Pseudocode for this method is provided in Algorithm \ref{alg:nonts}. In the standard setting, the PC algorithm is a constraint-based causal discovery algorithm for the causal sufficient case that relies on the Markov and Faithfulness assumption. In its first stage (skeleton discovery), adjacencies are learned based on iteratively testing conditionally independence of pairs of variables at some significance level $\alpha$. Afterwards, the links are oriented based on a set of rules. We will focus on how its skeleton phase needs to be adapted.
 \begin{enumerate}
     \item In the first step, we discover context-system links. We iteratively test independence between the following node pairs $(X^i, C^j)$, and $(C^j, X^i)$ for all $i,j$ while conditioning on subsets of the union over system and observed context nodes $\mathbf{X} \cup \mathbf{C}$. In other words, we initialize a fully connected graph between the system and context variables, eliminate the edges between context variables, and run the skeleton phase. By ignoring the dummy node in this step, we are able to circumvent the faithfulness violation that stems from the fact that every observed context node is a deterministic function of the dummy. 
     By the exogeneity of the context to the system (assumption \ref{ass:JCI}), we already know that any link between context variable $C$ and system variable $X$ is oriented as $C \rightarrow X$. Therefore, we construct the set of observed contextual parents $\text{Pa}_C(X^i)$ of each system variable $X^i$ from all observed context variables that are found to be adjacent to $X^i$.
    \item In the second step, we focus on the discovery of dummy-system links. In particular, we test independence between $D$ and each $X^i \in \mathbf{X}$ conditional on subsets of $\mathbf{X}$ and the found contextual parents $\text{Pa}_C(X^i)$. Combined with the expert knowledge that the dummy cannot be a descendent of a system variable, this gives us the dummy parents of $X^i$. We denote the set of dummy and contextual parents of $X^i$ by $\text{Pa}_{CD}(X^i)$.
    \item Finally, we run the skeleton phase of the full PC algorithm on $\mathbf{X} \cup \mathbf{C} \cup \{D\}$ while incorporating the background knowledge of the links from $\text{Pa}_{CD}(X^i)$ to $X^i$, and no context-context and context-dummy links.

 \end{enumerate}

Since context-system and dummy-system links are oriented by assumption, the orientation phase, see \cite{meek1995causal} for rules, needs to be applied to orient the system variables only. Note, however that we are taking between context-system or dummy-system edges into account whenever they form an unshielded triple with two system variables, i.e.\, $C \rightarrow X^i \oo X^j$ or $D \rightarrow X^i \oo X^j$. 
This allows to orient more edges than only considering triples of system variables.

\RestyleAlgo{ruled}
\begin{algorithm}
\caption{J-PC (for non-time series), pseudocode for poolData, and partialSkeletonPC is provided in SM}\label{alg:nonts}
\KwData{Background knowledge on context-system link orientation $\mathcal{E}$, observational data $(\mathbf{X}^{(m)})_{m=1, \ldots, M}$ in $M$ dataset, observed context variables $(\mathbf{C}^{(m)})_{m=1, \ldots, M}$ for each dataset, dummy variable $D$ with distinct values for each dataset, significance level $\alpha$}
\KwResult{graph $\mathcal{G}$}
$(\mathbf{X}, \mathbf{C}, D) \leftarrow \operatorname{poolData}((\mathbf{X}^{(m)}, \mathbf{C}^{(m)})_{m=1, \ldots, M}, D)$ \\
Set $\mathcal{P}_C := \{ (X, C), (C, X) | X \in \mathbf{X}, C \in \mathbf{C} \}$,\\ $\mathcal{P}_D := \{ (X, D) | X \in \mathbf{X}\}, \mathcal{P}_S := \{ (X, Y) | X, Y \in \mathbf{X} \}$ \\
Set $data_C := (\mathbf{X}, \mathbf{C})$ and $data_D := (\mathbf{X}, \mathbf{C}, D) $ \\
Set $\mathcal{C} = \emptyset$ \\
\For{index in [C, D]}{
    $\mathcal{G} \leftarrow \operatorname{partialSkeletonPC}(data_\text{index}, \alpha, \mathcal{P}_\text{index}, \mathcal{C})$ \\
    \For{$X$ in $\mathbf{X}$}{
        Orient context-system edge as in $\mathcal{E}$\\
        Add contextual parents $\text{Pa}_\text{index}(X)$ as in $\mathcal{G}$ to $\mathcal{C}$
    }
}
$\mathcal{G} \leftarrow \operatorname{partialSkeletonPC}(data_D, \alpha, \mathcal{P}_S, \mathcal{C})$ \\
Orient system edges using PC-orientation rules\\
\textbf{return} $\mathcal{G}$
\end{algorithm}

\paragraph{Time series case:}
\RestyleAlgo{ruled}
\begin{algorithm*}
\caption{J-PCMCI$^+$ (for time series), 
laggedSkeletonPCMCI$^+$ refers to Algorithm 1 in \cite{runge2020discovering}, partialContempSkeletonPCMCI$^+$ is a small adaption of Algorithm 2 of \cite{runge2020discovering} which is further described in the SM, colliderPhase and rulePhase refer to Algorithms 3 and 4 in \cite{runge2020discovering}}\label{alg:ts}
\KwData{Background knowledge on context-system link orientation $\mathcal{E}$, $M$ observational system time series datasets $\textbf{X}=(\textbf{X}^i)_{i \in \mathcal{I}}$, observed temporal context variables $\mathbf{C}_\text{time}$, observed spatial context variables $\mathbf{C}_\text{space}$, temporal and spatial dummy variables $ D_\text{time}, D_\text{space}$, significance level $\alpha$, maximal time lag $\tau_\text{max}$, $CI(X, Y, \mathbf{Z})$}
\KwResult{graph $\mathcal{G}$}
$\{\hat{\mathcal{B}}^-_t(X) | X \in \mathbf{X} \cup \mathbf{C}_\text{time} \} \leftarrow \text{laggedSkeletonPCMCI$^+$}(\mathbf{X} \cup \mathbf{C}_\text{time})$ \red{} \\
Initialize $\mathcal{C}(X) \leftarrow \emptyset$ for all $X \in \mathbf{X}$\\
Set 
$\mathcal{P}_C := \{ (C^k_{t-\tau}, X^i_t), (X_t^i, C_t^k) | \tau \geq 0, \forall i,k \}$, $\mathcal{P}_D := \{ (X_t^i, D), (D, X_t^i) | \text{for } D \in \{ D_\text{time}, D_\text{space}  \}, ~ \forall i \}$,  $\mathcal{P}_S = \{ ((X^j_{t-\tau}, X_t^i))_{\tau > 0}, (X_t^i, X_t^j) | i,j\}$\\
Set $data_C := \mathbf{X} \cup \mathbf{C}_\text{time} \cup \mathbf{C}_\text{space}$, and $data_D := \mathbf{X} \cup \mathbf{C}_\text{time} \cup \mathbf{C}_\text{space} \cup \{ D_\text{time}, D_\text{space}  \}$\\
\For{index in [C, D]}{
    $\mathcal{G}$ $\leftarrow$ partialContempSkeletonPCMCI$^+$($data_\text{index}$, $|\mathcal{I}|$, $\tau_\text{max}$, $\alpha$, $\hat{\mathcal{B}}^-_t(X)$, $\mathcal{C}(X)$, $\mathcal{P}_\text{index}$)\\
    \For{$X$ in $\mathbf{X}$}{
    orient context-system edge as in $\mathcal{E}$\\
    add all context nodes that are adjacent to $X$ in $\mathcal{G}$ to $\mathcal{C}(X)$\\}   
}

$\mathcal{G}$ $\leftarrow$ partialContempSkeletonPCMCI$^+$($data_D$, $|\mathcal{I}|$, $\tau_\text{max}$, $\alpha$,  $\hat{\mathcal{B}}^-_t(X)$, $\mathcal{C}$ $\mathcal{P}_S$)\\
$\mathcal{G},$ sepset, ambigious triples, conflicting links $\leftarrow$ colliderPhase($data_D$), i.e., on all unshielded triples $X_{t-\tau}^i \rightarrow X_{t}^k \oo X_{t}^j $ ($\tau > 0$) or $X_{t}^i \oo X_{t}^k \oo X_{t}^j $ or $K \rightarrow X_t^i \oo X_t^j$ with $K \in \mathbf{C}_\text{time} \cup \mathbf{C}_\text{space} \cup \{ D_\text{time}, D_\text{space} \}$\\
$\mathcal{G}$, conflicting links $\leftarrow$ rulePhase($\mathcal{G},$ ambigious triples, conflicting links)\\
\textbf{return} $\mathcal{G}$
\end{algorithm*}

Next, we combine the PCMCI$^+$ algorithm \cite{runge2020discovering} with observed context and dummy  variables. To recall briefly, PCMCI$^+$ is a causal discovery algorithm for time series data, that allows for both contemporaneous and lagged links and assumes causal sufficiency. It consists of two steps. The first \emph{$\text{PC}_1$ lagged phase} infers a superset of the lagged parents together with the parents of contemporaneous ancestors. Next, the \emph{MCI contemporaneous phase} starts with links found in the previous step and all possible contemporaneous links, it then conducts momentary conditional independence (MCI) with a modified conditioning set learned in the previous step to increase detection power. 

Our method consists of four main steps: one $\text{PC}_1$ lagged phase and three MCI phases. In the first step, supersets of the lagged parents of the system and observed temporal context nodes are discovered by running the $\text{PC}_1$ lagged phase on this subset of nodes.
Next, the MCI test is run on pairs of system and context nodes conditional on subsets of system and context, i.e.\ perform MCI tests for pairs $((C^j_{t-\tau}, X^i_t))_{\tau > 0}$,  $(C_t^j, X_t^i)$, $(X_t^i, C_t^j)$ for all $i,j$, 
    \[
        C_{t-\tau}^i \indep X_t^j | \mathbf{S}, \hat{\mathcal{B}}^-_t(X_t^j)  \setminus \{ C_{t-\tau}^i \}, \hat{\mathcal{B}}^-_{t-\tau}(C_{t-\tau}^i)
    \]
with $\mathbf{S}$ being a subset of the contemporaneous adjacencies $\mathcal{A}_t(X_t^j)$ and $\hat{\mathcal{B}}^-_t(X_t^j)$ are the lagged adjacencies from step one. If $C$ is a spatial context variable, we only have to test the contemporaneous pairs $(C_t^j, X_t^i)$, $(X_t^i, C_t^j)$ for all $i,j$. If $C_t^j$ and $X_t^i$ are conditionally independent, all lagged links between $C_t^j$ and $X^j_{t-\tau}$ are also removed for all $\tau$.
In the third step, MCI tests on all system-dummy pairs conditional on the superset of lagged links, the discovered contemporaneous context adjacencies, as well as on subsets of contemporaneous system links, are performed, i.e.\ test for $(D, X_t^i)$, $(X_t^i, D)$ for all $i$, i.e.\ 
    \[
    D \indep X_t^j | \mathbf{S}, \hat{\mathcal{B}}^C_t(X_t^j)
    \]
    where $\mathbf{S} \subset \mathcal{A}_t(X_t^i)$ and $\hat{\mathcal{B}}^C_t(X_t^j)$ are the lagged and contextual adjacencies found in the previous step.
If $D$ and $X_t^j$ are found to be conditionally independence, links between $D$ and $X^j_{t-\tau}$ are removed for all $\tau$.
Using assumption \ref{ass:JCI}, context node is the parent in all system-context links.
Finally, in the fourth step, we perform  MCI tests on all system pairs conditional on discovered lagged, context and dummy adjacencies, as well as on subsets of contemporaneous system links and orientation phase. In more detail, we perform MCI test for pairs $((X^j_{t-\tau}, X_t^i))_{\tau > 0}$, $(X_t^i, X_t^j)$ for all $i, j$, i.e.\ 
    \[
     X^i_{t-\tau} \indep X_t^j | \mathbf{S}, \hat{\mathcal{B}}^{CD}_t(X_t^j)  \setminus \{ X_{t-\tau}^i \},  \hat{\mathcal{B}}^{CD}_t(X_{t-\tau}^i) 
    \]
    where $\mathbf{S} \subset \mathcal{A}_t(X_t^i)$ and $\hat{\mathcal{B}}^{CD}_t(X_t^j)$ are the lagged, contextual, and dummy adjacencies found in the previous steps.
Finally, all remaining edges (without expert knowledge) are oriented using the PCMCI$^+$ orientation phase while making use of all triples involving one context or dummy variable and two system variables as in the non-time series case.

\subsection{Theoretical Results}\label{subsec:theorems}
Proofs for the following statements are provided in SM.

\begin{theorem}[Non-time series consistency result] \label{thm:nonts}
    Denote the output of J-PC (Algorithm \ref{alg:nonts}) as $\mathcal{G}_{alg}$.
    Under the assumptions \ref{ass:sufficiency},\ref{ass:JCI}, \ref{ass:no_conf},  \ref{ass:asymmetry}, and assuming consistent conditional independence tests are used, the dummy deletion of $\mathcal{G}_{alg}$ corresponds to the dummy-deleted ground truth graph as the number of data sets $M$ tends to infinity.
\end{theorem}

Note that here the dummy-deleted ground truth graph is the target graph (definition \ref{def:target_graph}) adapted to the non-time series case.

\begin{theorem}[Time series consistency result] \label{thm:ts}
    Denote the time series graph output of J-PCMCI$^+$ (Algorithm \ref{alg:ts}) as $\mathcal{G}_{alg}$.
    Under assumptions \ref{ass:sufficiency}, \ref{ass:JCI}, \ref{ass:no_conf}, \ref{ass:asymmetry}, and assuming consistent conditional independence tests are used, the dummy deletion of $\mathcal{G}_{alg}$ corresponds to the target graph (definition \ref{def:target_graph}) as the number of data sets $M$ and the number of times steps $T$ tend to infinity.
\end{theorem}

The following consequence of theorem \ref{thm:ts} allows us to relax the rather strong assumption that latent context variables cannot mediate or confound an observed context variable and a system variable.

\begin{corollary}
    If some of the observed context variables are treated as unobserved, and the assumptions \ref{ass:sufficiency}-\ref{ass:asymmetry} still hold, our method J-PCMCI$^+$ will recover the correct system-system adjacencies.
\end{corollary}

In particular, even if all context variables are treated as unobserved, our algorithm yields the correct induced graph over the system variables.

\section{Numerical Experiments}\label{sec:num}
\paragraph{Data simulation}
We generate toy data from the SCM \ref{eq:scm} where we assume the functions $f_i, g_k$, and $h_l$ to be linear. We also evaluate the method on data where the mechanisms $f_i$ are nonlinear. For a more detailed description of this setup, see the SM. In particular, in the linear setting, for system variables $\mathbf{X}=\{X^i\}_{i\in \mathcal{I}}$ and temporal context variables $(C^{\text{time},k})_{k \in \mathcal{K}\text{time}}$ and spatial context variables $(C^{\text{space},l})_{l \in \mathcal{K}\text{space}}$, we consider the following ground truth SCM $ X^{i,m}_t = a_i X^{i,m}_{t-1} + \sum_{j} b_j X^{j,m}_{t-\tau_j} + \sum_{j} c_jC^{\text{time},j}_{t-\tau_j} + \sum_{j} d_jC^{j, m}_\text{space} + \eta^{i,m}_t$,
where $i \in \mathcal{I}$, $t= 1, \ldots, T$, and $m=1, \ldots, M$, $C^{\text{time},k} \sim \mathcal{N}(0,1)$, $C^{\text{space}} \sim \mathcal{N}(0,1)$. Furthermore, $\eta^i \sim \mathcal{N}(0,1)$ i.i.d., $a_i$ autocorrelation parameter uniformly drawn from $[0.3, 0.8]$, coefficients $b_j,c_j, d_j$ are uniformly drawn from $[0.5, 0.9]$, $50\%$ of the links are contemporaneous, the remaining lags are drawn uniformly from $[1,3]$. After the data has been generated, its variance is rescaled to one across all datasets to avoid varsortability \cite{reisach2021beware}. In the numerical experiments, we make the restriction that one system node can have at most one contextual parent.
After the time series for the ground truth model has been generated a certain fraction (indicated by parameter frac\_observed) of the context nodes is selected to be observed, the others are unobserved.
Note that in our simulated data all context nodes are exogeneous (they do not even have other context nodes as parents). We decided to set the experiments up in this way to put the focus on the discovery of the system-context links and also the deconfounding property of the dummy nodes.

\paragraph{Setup}
We evaluate the performance of our method using True (TPR) and False Positive Rate (FPR) for the adjacencies which is calculated only on the system-system context links. Separately, we report the TPR and FPR on system-observed context links. All metrics and their standard deviations are computed on the estimated graphs of $50$ realizations of the model from time series with length $T$.

We compare our method to PCMCI$^+$ run on the data of the system nodes only by simply concatenating the data, as well as to PCMCI$^+$ where we have only included the observed (context and system) nodes. In this variant, we took care to include spatial context nodes only once in the time series graph. We build upon the implementation of PCMCI$^+$ algorithm within the Tigramite software package \citep{runge2019detecting} published under the GNU General Public License.

 We used the following model parameters in our experiments: Number of system nodes $|\mathcal{I}|=5$, number of context nodes $|\mathcal{K}_\text{time}| + |\mathcal{K}_\text{space}| = 3$ 
 , maximal time lag $\tau_{\operatorname{max}}=2$, significance level $\alpha=0.05$. We use an extension of the (component-wise) partial correlation conditional independence test.
We vary the value of the time sample size $T$, and number of datasets $M$. Results for other fractions of observed context nodes can be found in the SM.

\paragraph{Benefit of including expressive context nodes and dummy}
In figure \ref{fig:results_cnt}, we report the TPR and FPR for observed context-system links, respectively. We see that our method finds the links between all observed context nodes and the system nodes as well as the PCMCI$^+$ variant with observed context nodes included in the graph. In figure \ref{fig:results}, we observe that the performance of our method is comparable to only including dummy variables when evaluated only on the system-system links. These two observations illustrate the benefit of using our method: If we would only rely on the dummy context variables, we would achieve deconfounding of the system variables but, naturally, we would not find any of the links between the observed context nodes and the system. In other words, we would not be able to learn which parts of the system are dependent on the properties of the context. On the other hand, if we would only include the observed context nodes, we would not be able to remove the confounding that could be introduced by the latent context variables. The effects of the counfounding latent context variables, are visible in the rise in FPR on system-system links in figure \ref{fig:results}. We also see lower edgemark recall and precision when only using system data, see the figures in the SM.

\begin{figure}
    \centering
    \includegraphics[width=\linewidth]{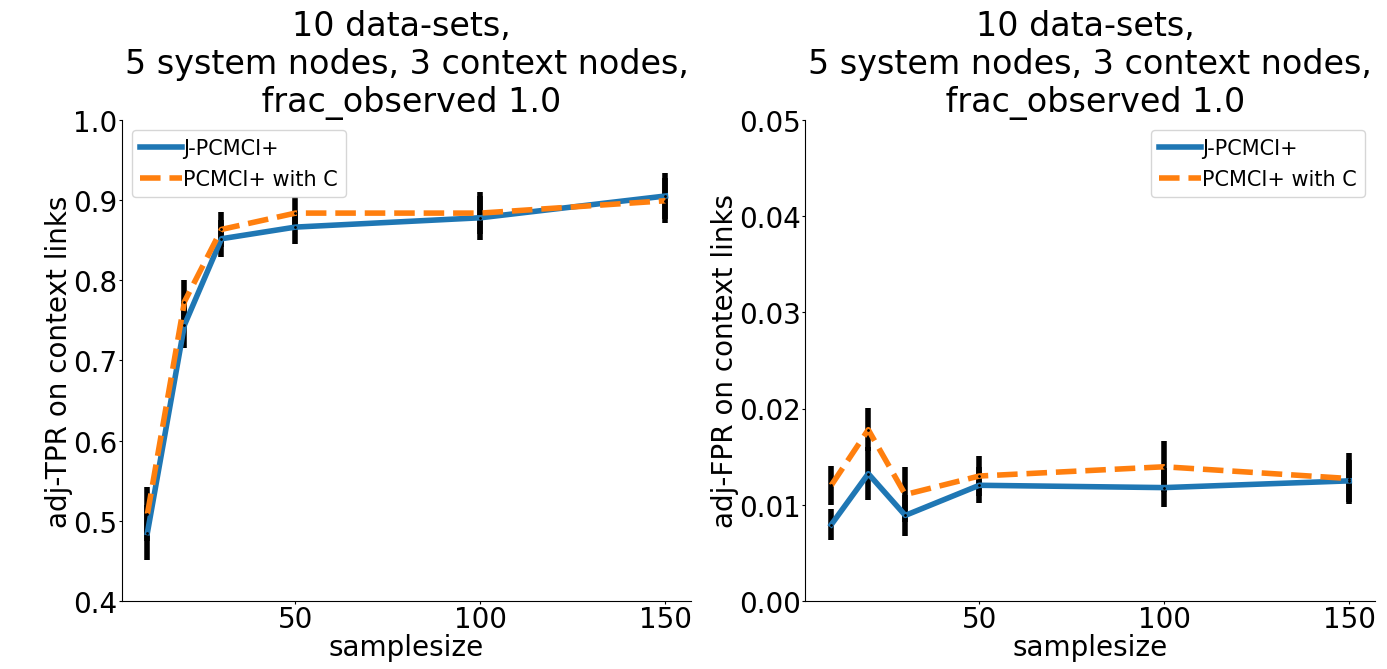}
    \caption{Discovery results of context-system links for varying sample sizes $T$, $M=10$. All other setup parameters are set as the defaults described in the main text. In this setting all context nodes are observed. Here, we compare our method (J-PCMCI$^+$) to PCMCI$^+$ using all data of observed nodes. Note that the maximal TPR that can be reached is equal to frac\_observed.}
    \label{fig:results_cnt}
\end{figure}

\paragraph{Convergence analysis (time and space dimension)}
We want to numerically study the finite sample properties of our method. For that, we look at the TPR and FPR on the system links for varying time sample size $T$ while keeping the number of spatial contexts $M$ fixed, and the other way around, see figure \ref{fig:results}. In these experiments half of the context nodes are unobserved. In the SM, we have also included 3D-plots for TPR and FPR on all pairs of $T$ and $M$.

\begin{figure}[htb!]
    \centering
    \includegraphics[width=\linewidth]{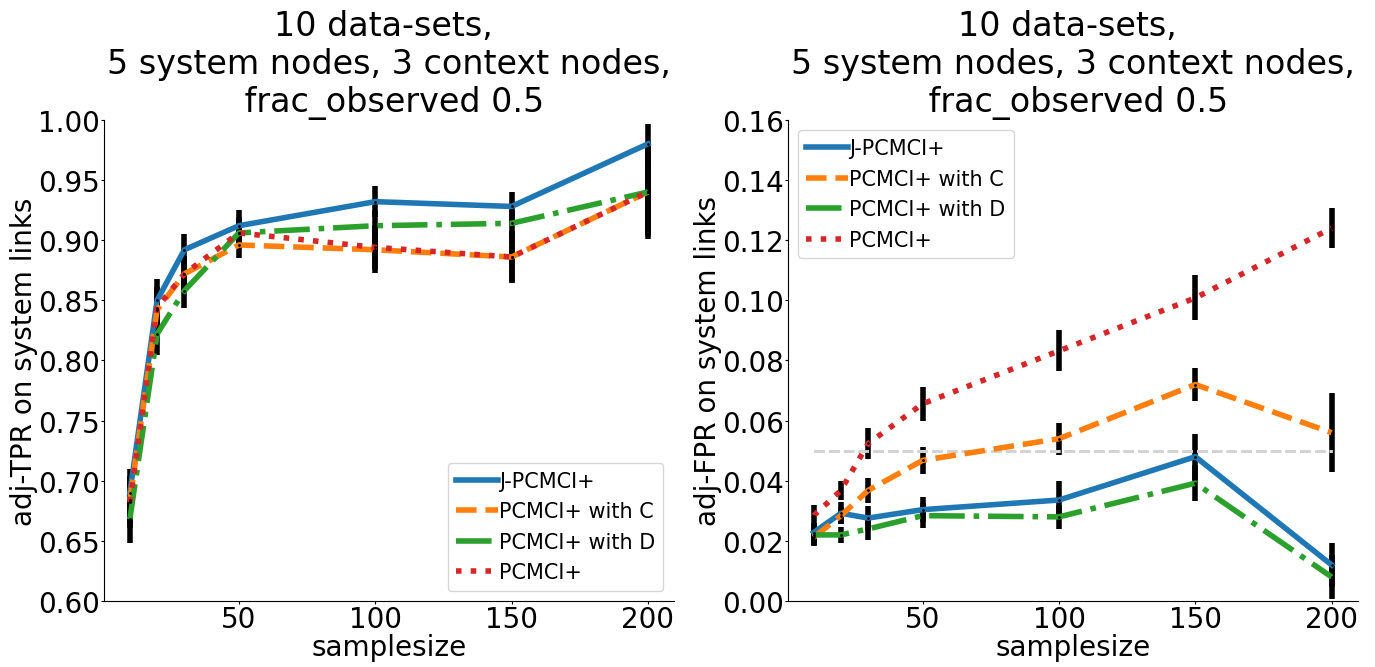}
    \includegraphics[width=\linewidth]{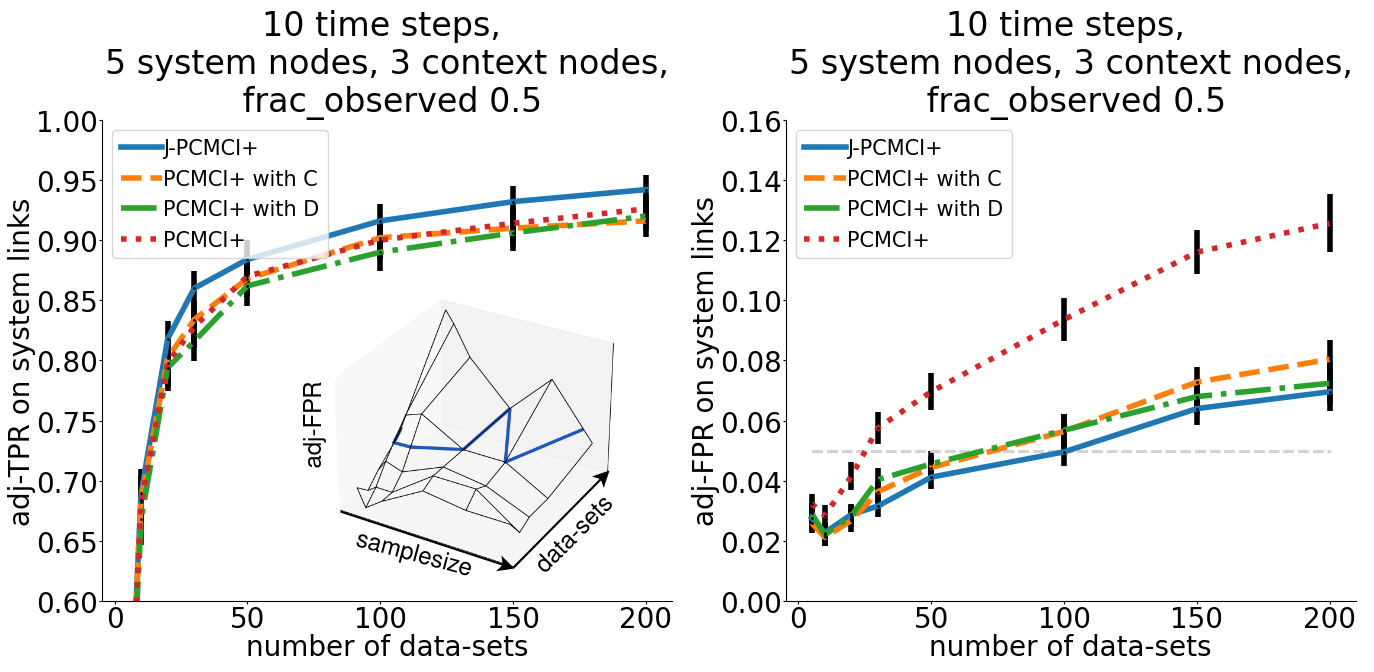}
    \caption{Discovery results of system-system links for varying sample sizes $T$, and fixed $M=10$ (top row), and varying number of contexts $M$, and fixed $T=10$ (bottom row). All other setup parameters are set as the defaults described in the main text. In this setting half of the context nodes are observed. We compare our method (J-PCMCI$^+$) to PCMCI$^+$ using all data of observed nodes (PCMCI$^+$ with C), using all data of system variables and including dummies (PCMCI$^+$ with D), and only using data of system variables (PCMCI$^+$). 
    The inset shows the adj-FPR-surface with the contour of the $\alpha$-level in a simplified experimental setup to visualize the convergence of the method. Refer to the SM for details. 
    }
    \label{fig:results}
\end{figure}
It is a well established result in the econometrics literature, that, when considering fixed effects models, a bias is introduced in the OLS estimator of the slope parameter \cite{nickell1981biases}. A similar inconsistency problem can be observed in our method whenever $T$ is kept fixed and is small compared to $M$: Even when the number of spatial contexts $M$ goes to infinity, the links are not discovered correctly, see figure \ref{fig:results} (in particular the FPR plots) and the SM. 

\section{Discussion and Conclusions}\label{sec:discussion}

We presented an algorithm (J-PCMCI$^+$) for causal discovery from a collection of multivariate time series datasets that is able to deal with observed and unobserved context variables underlying the datasets.  We established its asymptotic consistency and studied its convergence properties numerically. 

The \textbf{main strengths} of J-PCMCI$^+$ are that it combines the efficient algorithm from \citet{runge2020discovering} (handling highly autocorrelated time series) with ideas in \citet{mooij2020joint,huang2020causal} to model observed as well as unobserved contexts. Pooling data from multiple datasets and adding observed contexts as well as dummies has several important benefits: (1) pooling increases sample size, (2) adding observed contexts and dummies allows to learn context-system relations and can help to orient  system-system links, (3) dummy variables allow to remove confounding at least from latent context variables.  J-PCMCI$^+$ inherits the benefits of PCMCI$^+$ for high-recall and accounting for autocorrelation in the conditional independence tests. We find numerically that J-PCMCI$^+$ has good performance for sufficiently large sample sizes and moderate numbers of datasets.

The \textbf{main weaknesses} are that some assumptions might be strong and unrealistic. In particular, JCI assumption 2 and the assumption that prohibits latent context mediation of the observed context to system link can be hard to justify depending on the setup, see Assumption \ref{ass:JCI}. In the SM, we discuss ramifications of partially relaxing this assumption. Our numerical experiments indicate that for too small sample sizes, we get inflated false positives due to missing latent confounders. We also cannot overcome the fundamental finite-sample bias in the OLS estimator~\citep{nickell1981biases} for small sample sizes and dummy variables carry the disavantage of increasing dimensionality.

In \textbf{future work}, we plan to extend the method to weaken the partially strong assumptions and to account for latent system variable confounding \cite{gerhardus2020high}. Furthermore, while our method only uses high-dimensional dummy variables where it finds dependence, we haven't looked yet into adapting CI tests specifically for dummy variables. In this work, we also only considered acyclic SCMs. However, an extension to include the cyclic case is possible based on the consistency of the PC algorithm in the cyclic setting \cite{mooij2020constraint}. Moreover, an in-depth analysis of finite sample properties of the presented method is needed.

\begin{acknowledgements}
W.G. was supported by the Helmholtz AI project CausalFlood (grant no. ZT-I-PF-5-11). 
J.R. and U.N.  received funding from the European Research Council (ERC) Starting Grant CausalEarth under the European Union’s Horizon 2020 research and innovation program (Grant Agreement No. 948112).
This work used resources of the Deutsches Klimarechenzentrum (DKRZ) granted by its Scientific Steering Committee (WLA) under project ID bd1083.

We thank the anonymous reviewers for their helpful comments.  
\end{acknowledgements}

\bibliography{gunther_390}
\end{document}


\onecolumn 
\maketitle

In this Supplementary Material, we provide some subtleties of the target graph of J-PCMCI$^+$, the dummy-projection and deletion operations, the embedding and representation of the dummy variable and how to relax some of the assumptions. Furthermore, we give background on the challenge of determinism within causal discovery, illustration on how context nodes can help with orienting additional edges, proofs for the main theorems, additional pseudocode and details on the simplified experimental setup as well as additional plots for the numerical experiments.

\appendix

\section{Code}
The code to reproduce the experimental results can be found under the following url \url{https://github.com/guenwi/J-PCMCIplus}. The method (J-PCMCI$^+$) will also be made available as part of the tigramite package (\url{https://github.com/jakobrunge/tigramite}).

\section{More on dummy-projection and -deletion}

\subsection{The target graph of J-PCMCI+}

 We define the "target graph" as our ultimate object of interest, which is the causal graph between the system nodes. 
By extension, it is implied that the links between the context nodes as well as those between the dummy and system nodes aren't of interest, the latter additionally so because the dummy variable is not a causal variable. 
To make this more tangible, we provide additional illustration of the target graph in relation to the dummy-projection and dummy-deleted version of the ground truth graph in figure \ref{fig:target_graph}.

 \begin{figure}
     \centering
     \includegraphics[width=9cm]{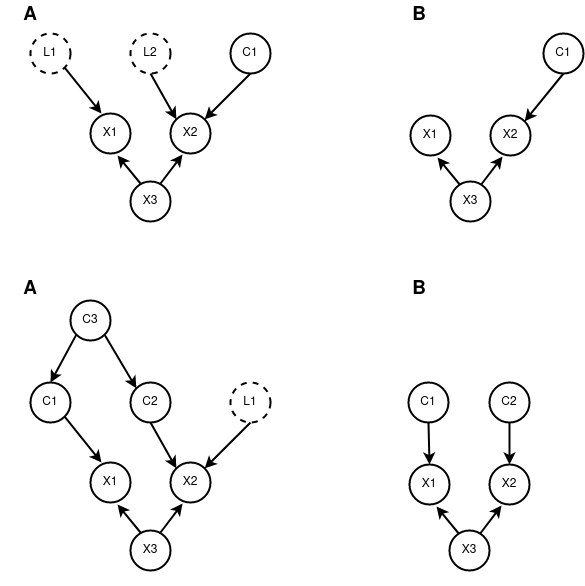}
     \caption{Visualization of the summary graphs of SCM~(1) (A), as well as the corresponding target graph (B). The node $D$ in the dummy projected graph can be either the time dummy $D_\text{time}$ or the space dummy $D_\text{space}$. Latent context nodes are visualized using dashed circles, dashed arrows denote deterministic dependencies (not part of dummy projection).}
     \label{fig:target_graph}
 \end{figure}

\subsection{Dummy confounding}
A misleading fact about the dummy projection, and also of $\mathcal{G}_{alg}$, which is the result of algorithm J-PCMCI$^+$ (J-PC, respectively),  is that it can contain system variables confounded by the dummy, that do not correspond to actual latent confounding, see figure \ref{fig:dummy_proj} for a visualization of one such case. However, as we prove in Section 4.3, this is not a concern because we are interested in the true causal graph over the system variables together with edges from context to system variables, and for this task conditioning on such a dummy that isn't a true confounder doesn't lead to wrong inferences.

Furthermore, note that we include the time-dummy $D_\text{time}$ only once into the time series graph. Since the time-dummy does not contain information on the specific value of the unobserved context variables but only encodes expert knowledge on their structure, we are not able to discover at which specific lag the causal relationship between the latent temporal context variables and the system variables occurs. However, we are able to find whether the system variables are influenced by a temporal context variable or not.
See figure \ref{fig:ts_graph} for a visualization.

 \begin{figure}
     \centering
     \includegraphics[width=0.4\linewidth]{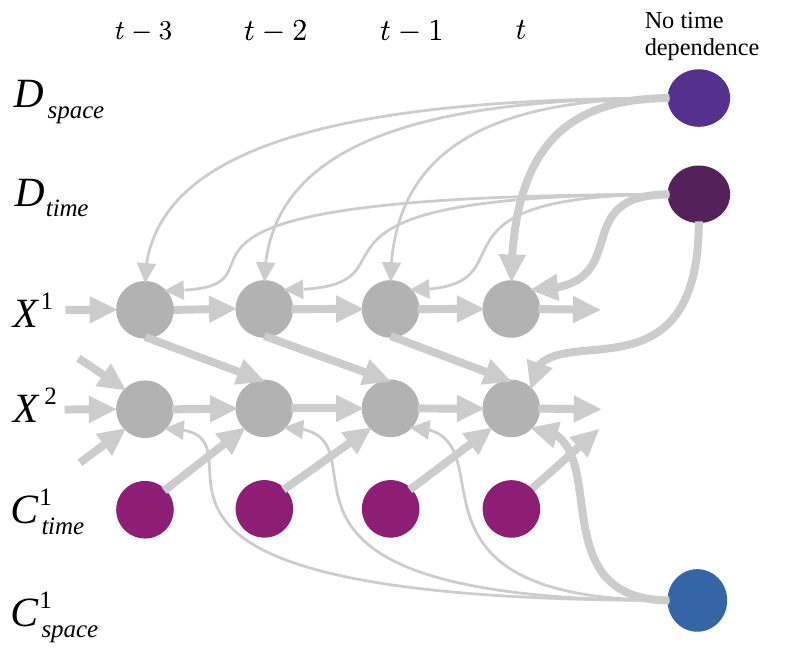}
     \caption{Unrolled time series graph with system variables $X^1$, $X^2$, an observed temporal context variable $C^1_{time}$ with (possibly lagged) links to system variables, an observed spatial context variable $C^1_{space}$ as a single node as it's constant over time,
     space and time dummy $D_{\rm space} $ and $D_{\rm time} $ as single nodes.
     Note that this graph omits links between context variables and between dummy and context variables.
     For better readability, we used thinner arrows for lagged links from the context and dummy variables to the system variables.}
     \label{fig:ts_graph}
 \end{figure}

\subsection{Embedding of the dummy variables}
Here, we provide further detail on how to represent or encode the dummy variable.
The choice of embedding matters most for testing conditional independence between dummy and system variables.
Potential choices of embeddings include one-hot encoding, which we use, or using the integers that denote the time or data set index directly (as done in \citep{huang2020causal}), among others.

In CD-NOD \citep{huang2020causal}, the auxiliary variable corresponds to the domain or time index, i.e., the dummy takes values in $\{1, \ldots, n_C\}$, where $n_C$ is the number of contexts. 
To then be able to test for marginal and conditional independence between a system variable $X$ and the dummy $D$, \citet{huang2020causal} employ the KCI test \citep{zhang2011kernel} since the functional relationship between $D$ and $X$ is highly non-linear. In case of non-stationary data, they also assume that all temporal context variables are a smooth function of the time-dummy, thus they also need to keep the time order in their embedding of the dummy (which we do not do).
When using a one-hot encoded dummy in combination with a partial correlation test for testing whether a system variable depends on the context, we are essentially testing for differences in mean of that system variable as the context changes, since the partial correlation coefficient reduces to the point-biserial correlation coefficient \cite{sheskin2020handbook} if one of the variables is dichotomous. If we would adapt the CI test, the same could be achieved with an integer-embedding of the dummy.

Furthermore, the choice of embedding also has implications on how easy it is to regress out context information from the system variables when testing system-system adjacencies conditional on the dummy. Using the one-hot encoding of the dummy values, we are centering the system data within each dataset or across time. This is very related to the well-established technique of fixed effects panel regression.

\subsection{Relaxing the No-Mediation Assumption} 
It is possible to relax the no-mediation assumption, which is part of Assumption 2. However, this will result in the context-system links as discovered by J-PCMCI$^+$ representing ancestral rather than direct causal relationships. A consequence of that is that the graph $\mathcal{G}_{alg}$ that corresponds to the (time-series) graph resulting from algorithm J-PCMCI$^+$ is not identical to the dummy-projected graph $\mathcal{G}_D$ of the ground truth graph $\mathcal{G}$. In particular $\mathcal{G}_{alg}$ will have more system-context links than those in $\mathcal{G}_D$. Such link could appear because of an observed context variable that is indeed a parent of the system variable, i.e., links that also appear in $\mathcal{G}_D$. However, it could also happen that a latent context $L$ is a mediator between an observed context variable and a system variable. Since we cannot condition on the dummy in the first step of the algorithm (due to the deterministic relationship between dummy and context), the algorithm will not remove this link. However, it is not contained in the dummy projected ground truth graph because there is no actual link to system from the observed context variable.

Consequently, the consistency theorem \ref{thm:ts} does no longer hold since it relies on the dummy-projection of the ground truth graph which does not include ancestral links between context and system variables that are mediated by a latent context.
The theorem would be adapted to a new definition of the dummy-projection that includes such links.

\begin{figure}
    \centering
    \includegraphics[width=13cm]{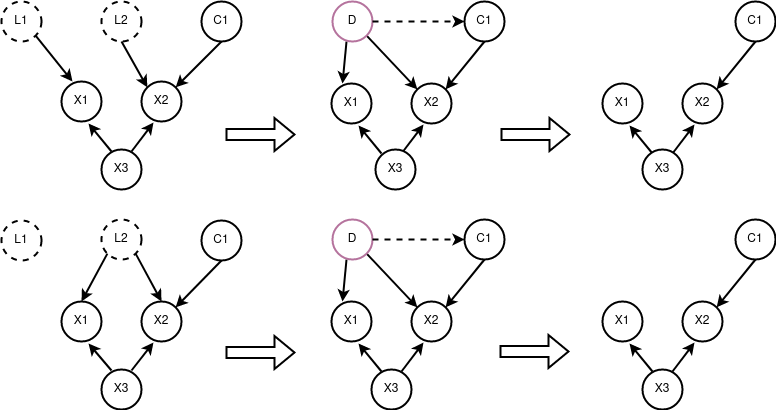}
    \caption{Visualization of the dummy projection operator (middle) and the dummy deletion (right) on summary graphs of SCM~(1). The node $D$ can be either the time dummy $D_\text{time}$ or the space dummy $D_\text{space}$. Latent context nodes are visualized using dashed circles, dashed arrows denote deterministic dependencies (not part of dummy projection). The first row indicates that a dummy confounder does not correspond to a real latent confounder, while in the second row it does.}
    \label{fig:dummy_proj}
\end{figure}

\section{Faithfulness violation due to determinism}

Causal inference on data containing deterministic relationships and the challenges thereof have been dealt with previously, e.g.\ \citep{daniusis2012inferring}, \citep{lemeire2011inferring}. 
Let us look at an example which was taken from \citet{lemeire2012conservative} to illustrate the challenge that is introduced by deterministic relationships.

Let $ X \rightarrow Y \rightarrow Z$ be the ground truth causal graph over the variables $X$, $Y$, $Z$, where there is a deterministic relation between $Y$ and $X$, i.e.\ there exists a function $f(\cdot)$ s.t. $Y = f (X)$. Therefore, $Y \indep Z | X$ since $X$ contains all information about $Y$.

This illustrates that deterministic relations generate additional independencies beyond those implied by the Markov condition. In other words, the true DAG is not faithful to the joint probability distribution of $X, Y, Z$. To describe these additional independencies, the D-separation criterion has been introduced \citep{geiger1990identifying}.
It is worth noting that the PC algorithm will not be sound when deterministic relationships are present. In the above example, the algorithm will remove both the edge between $Y$ and $Z$, and also between $X$ and $Z$.
This will happen because the conditional independence $Y \indep Z | X$ (wrongly) suggests that $X$ separates $Y$ and $Z$. On the other hand, $X \indep Z | Y$ due to d-separation.

\section{Context nodes help in orienting edges}
We quickly recap how context variables help in orienting additional system-system links. We consider the situation where $C \rightarrow X - Y$ and there is no edge between $C$ and $Y$. Then $C \rightarrow X - Y$ forms an unshielded triple. Then we can use standard collider orientation rules to orient the edge between $X$ and $Y$. In more detail:
\begin{enumerate}
    \item[(i)] If $Y$ and $C$ are independent given a set of variables that does not include $X$, then the triple is a V-structure, and we have $X \leftarrow Y$.
    \item[(ii)] Otherwise, if $Y$ and $C$ are independent given a set of variables including $X$, then we have $X \rightarrow Y$.
\end{enumerate}
See figure \ref{fig:orientations} for a visualization.

\begin{figure}
    \centering
    \includegraphics[width=7cm]{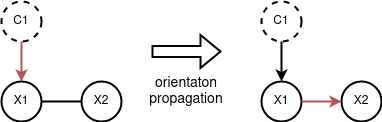}
    \caption{Visualization of how the context variables help orient additional edges by making use of the assumption that context nodes are exogenous to the system (left) and the standard rules of orientation propagation (right).}
    \label{fig:orientations}
\end{figure}

\section{Proofs}
\subsection{Proof of Theorem \ref{thm:nonts}}
Before we get to the proof of Theorem \ref{thm:nonts}, we note that the following useful lemma holds for the non-time series case.
\begin{lemma} \label{eq:dummy_cond}
For two system variables $X$ and $Y$, it holds for any $S \subset \mathbf{X}$
\begin{equation*} 
    X \indep Y | S \cup \{ D \} \quad \iff \quad X \indep Y | S \cup \mathbf{C} \cup \mathbf{L}.
\end{equation*}
For a definition of the sets $\mathbf{X}$, $\mathbf{C}$ and $\mathbf{L}$, please refer to the main text. Note that in the non-time series case, there is no time dummy, therefore $D := D_{\rm space} $.
\end{lemma}
\begin{proof}
This was already shown in the proof of Theorem 1 by \citet{huang2020causal} but, for convenience, we repeat the arguments here.
All system variables can be expressed as functions of $\mathbf{C}$, $\mathbf{L}$, and the noise. Therefore, the conditional distribution of the system given the dummy (i.e.\ the distribution within each data set) $P (\mathbf{X}|D)$ is determined by the joint distribution of the noise, and the observed and latent context variables $\mathbf{C} \cup \mathbf{L}$. This implies $P(X, Y | S \cup \mathbf{C} \cup \mathbf{L} \cup \{D\}) = P(X, Y | S \cup \mathbf{C} \cup \mathbf{L})$ where $S \subset \mathbf{X}$ (since the noise is independent of $D$). 
Then by recalling the weak union property of conditional independece as well using the fact that $\mathbf{L}$ and $\mathbf{C}$ are deterministic functions of $D$, it follows that, $X \indep Y | S \cup \{D\}$, i.e.\ $P(X, Y | S \cup  \{D\}) = P(X | S \cup  \{D\})P(Y | S \cup  \{D\})$ is equivalent to $P(X, Y | S \cup \mathbf{C} \cup \mathbf{L}) = P(X | S \cup \mathbf{C} \cup \mathbf{L})P(Y | S \cup \mathbf{C} \cup \mathbf{L})$.    

\end{proof}

We now recall theorem \ref{thm:nonts}.
\begin{theorem}[Non-time series consistency result] \label{thm:nonts}
Denote the output of J-PC (Algorithm 1 in the main text) as $\mathcal{G}_{alg}$.
Under the assumptions 1, 2, 3, 4, and assuming consistent conditional independence tests are used, the dummy deletion of $\mathcal{G}_{alg}$ corresponds to the dummy-deleted ground truth graph as the number of data sets $M$ tends to infinity.
\end{theorem}

\begin{proof}
Let us denote the skeleton of the projected ground truth graph with deleted dummy nodes by $ \mathcal{G}^*$. Similarly, we denote the skeleton of the dummy-deleted output of the algorithm by $ \hat{\mathcal{G}}^*$. We call their dummy-projected version of the ground truth graph $\mathcal{G}$, and the output of the algorithm (which is essentially a dummy projection) $\mathcal{G}_{alg}$.

First, we prove soundness of the algorithm, in other words we need to show that $\hat{\mathcal{G}}^* = \mathcal{G}^*$.

The soundness of context-system links follows from the soundness of the PC algorithm on the subset of system and observed context nodes. 
Let $X \in \mathbf{X}$ and $C \in \mathbf{C}$. The algorithm removes a link iff $X \indep C | S$ where $S \subset \mathbf{X} \cup \mathbf{C}$. Then Faithfulness (w.r.t. ground truth graph to $P_m(X,C,L)$) implies that all links not in $ \hat{\mathcal{G}}^*$ are also not in $ \mathcal{G}^*$. \\
We also need to show that any context-system links that are not in $ \mathcal{G}^*$ are also not in $ \hat{\mathcal{G}^*}$. If the link between $X$ and $C$ is not in $ \mathcal{G}^*$, then $X \indep C | S$ where $S \subset \mathbf{X} \cup \mathbf{C} \cup \mathbf{L}$. Using the assumption that system and context (Assumption 3) are not confounded by latent variables and latent context nodes cannot be mediators between system and context, this is equivalent to $X \indep C | S$ where $S \subset \mathbf{X} \cup \mathbf{C}$. This is tested at some iterative step of the PC-algorithm, and consequently the link is removed.

Since the dummy-deleted graphs do not contain any links to the dummy, we need to show soundness for the system-system links.
However, within that it is needed that we find the correct dummy-system links.
In other words we first show that if the link $D - X$ is not in $\mathcal{G}_{alg}$, then it also is not in $\mathcal{G}$.\\
If the link $D - X$ is not in $\mathcal{G}_{alg}$, then $X \indep D | S \cup \text{Pa}_{C}(X)$ where $S \subset \mathbf{X}$. This implies that for all latent context nodes $L$ holds $X \indep L | S \cup \text{Pa}_{C}(X)$ since $L$ can be expressed as a function of $D$, i.e.\ there exists a function $g$ with $L=g(D)$. Therefore by Faithfulness and the non-invertibility of the function $g$, there is also no link between $X$ and $L$ in the ground truth graph, and thus also no link to the dummy in its projected version $\mathcal{G}$. \\ 
For the other direction, we need to show that if the link $D - X$ is not in $\mathcal{G}$, then it also is not in $\mathcal{G}_{alg}$.
If the link $D - X$ is not in $\mathcal{G}$, then for all latent nodes $L$ it holds $L - X$ is not in the ground truth graph. By the Causal Markov Condition, it holds $X \indep L | \text{Pa}(X)$ for all $L$. This also implies that $\text{Pa}_L(X)$ is empty. We also know that $X$ can be expressed as a function of the context nodes $\mathbf{C}$, $\mathbf{L}$ and the noise. This means, conditional on $\text{Pa}(X)$, $X$ only depends on the noise. The noise is independent of $D$, thus $X \indep D | \text{Pa}(X)$. Therefore, the algorithm will remove this dummy-system link.

Now, we prove soundness for the system-system links. The algorithm removes a link iff $X \indep Y | S \cup \text{Pa}_{CD}(X)$ where $S \subset \mathbf{X}$, $\text{Pa}_{CD}(X)$ dummy and contextual parents of $X$.
Note that by Assumption 4 there exist functions $g^i, h^j$ s.t. $L^i = g^i(D)$ and $C^j = g^j(D)$, and thus $P(X, Y | S \cup \mathbf{C} \cup \mathbf{L} \cup \{D\}) = P(X, Y | S \cup \{D\})$.\\
So, if $D \in \text{Pa}_{CD}(X)$ this yields, together with Lemma \ref{eq:dummy_cond}
\[
X \indep Y | S \cup \text{Pa}_{CD}(X) ~ \implies ~ X \indep Y | S \cup \{D\} ~ \implies ~ X \indep Y | S \cup \mathbf{C} \cup \mathbf{L}.
\]
And Faithfulness (of ground truth graph to $P_m(X,C,L)$) implies that this link is not in $ \mathcal{G}^*$.
If  $D \not \in \text{Pa}_{CD}(X)$, Faithfulness is directly applicable.

It remains to show that system-system links not in $ \mathcal{G}^*$ are also not in $ \hat{\mathcal{G}^*}$. If the link between X and Y is not in $ \mathcal{G}^*$, then $X \indep Y | S$ where $S \subset \mathbf{X} \cup \mathbf{C} \cup \mathbf{L}$, and also $X \indep Y | S \cup \text{Pa}_C(X) \cup \text{Pa}_L(X)$ where $S \subset \mathbf{X}$.
Again, we distinguish two cases:\\
First note that by what we proved above $\text{Pa}_{CD}(X)$ in the dummy-projection $\mathcal{G}$ is a subset of $\text{Pa}_{CD}(X)$ in the dummy-projection $\mathcal{G}_{alg}$ (i.e. it might not contain the dummy if $\mathcal{G}_{alg}$ has a dummy link to X).
If $D \not \in \text{Pa}_{CD}(X)$ in $\mathcal{G}$, then there exists $V \in \mathbf{C} \cup \mathbf{X}$ with $X \indep L | V$ (since $L$ is a function of $D$), i.e.\ $\text{Pa}_{L}(X) = \emptyset$, and thus $X \indep Y | S \cup \text{Pa}_{C}(X)$, but also $X \indep Y | S \cup \text{Pa}_{C}(X) \cup \{D\}$. So, in any case (if $D$ is a parent of $X$ in $\mathcal{G}_{alg}$ or not) the algorithm removes the link.\\
On the other hand, if $D \in \text{Pa}_{CD}(X)$ in $\mathcal{G}$, we use that $X \indep Y | S \cup \text{Pa}_C(X) \cup \text{Pa}_L(X)$ is equivalent to $X \indep Y | S \cup \mathbf{C} \cup \mathbf{L}$ which is equivalent to $X \indep Y | S \cup \{ D \}$. If this holds, then also $X \indep Y | S \cup \{ D \} \cup \text{Pa}_{C}(X)$, which implies $X \indep Y | S \cup \text{Pa}_{CD}(X)$. So, also in this case the algorithm will remove the link. This concludes the soundness proof.

Completeness follows from soundness of the context-system, dummy-system and system-system links proved above, and the completeness of the PC-algorithm under tiered background knowledge \cite{andrews2020completeness}.
\end{proof}

\subsection{Proof of Theorem \ref{thm:ts}}
We extend Lemma \ref{eq:dummy_cond} to the time series case.
\begin{lemma} \label{eq:dummy_cond_ts}
Let $\mathbf{X}$ be time series data. Define $\mathcal{B}^-_{XY}:=(\hat{\mathcal{B}}^-_t(Y_t) \setminus \{ X_{t-\tau} \}), \hat{\mathcal{B}}^-_{t-\tau}(X_{t-\tau})$ where $\hat{\mathcal{B}}^-_t(X_t)$ denotes the lagged adjacency set resulting from the lagged skeleton phase of PCMCI$^+$ (Algorithm 1 in \cite{runge2020discovering}).
For two system variables $X_{t-\tau}$ and $Y_t$, it holds for any $S \subset \mathbf{X}$
\begin{equation*} 
    X_{t-\tau} \indep Y_t | S, \mathcal{B}^-_{XY} ,D_\text{time} , D_\text{space} \quad \iff \quad X_{t-\tau} \indep Y_t | S, \mathcal{B}^-_{XY}, \mathbf{C}, \mathbf{L},
\end{equation*}
and 
\begin{equation*} 
    X_{t-\tau} \indep Y_t | S, \mathcal{B}^-_{XY}, D_\text{space} \quad \iff \quad X_{t-\tau} \indep Y_t | S \cup \mathbf{C}_\text{space},\mathcal{B}^-_{XY}, \mathbf{L}_\text{space},
\end{equation*}
as well as
\begin{equation*} 
    X_{t-\tau} \indep Y_t | S, \mathcal{B}^-_{XY}, D_\text{time} \quad \iff \quad X_{t-\tau} \indep Y_t | S , \mathcal{B}^-_{XY}, \mathbf{C}_\text{time} , \mathbf{L}_\text{time}.
\end{equation*}
\end{lemma}

\begin{proof}
The following equation follows exactly in the same way as Lemma \ref{eq:dummy_cond}. Since the observed and latent context variables are either space- or time-dependent, this also works for the space and time dimension separately.
\end{proof}

\begin{theorem}[Time series consistency result] \label{thm:ts}
Denote the time series graph output of J-PCMCI$^+$ (Algorithm 2 in the main text) as $\mathcal{G}_{alg}$.
Under assumptions 1, 2, 3, 4, and assuming consistent conditional independence tests are used, the dummy deletion of $\mathcal{G}_{alg}$ corresponds to the target graph (definition 1) as the number of data sets $M$ and the number of times steps $T$ tend to infinity.
\end{theorem}

\begin{proof}
Let us denote the skeleton of the projected ground truth time series graph with deleted dummy nodes by $ \mathcal{G}^*$. Similarly, we denote the skeleton of the dummy-deleted time series graph output of the algorithm by $ \hat{\mathcal{G}}^*$. We call their dummy-projected version of the ground truth graph $\mathcal{G}$, and the output of the algorithm (which is essentially a dummy projection) $\mathcal{G}_{alg}$.

Soundness:\\
First, note that the lagged phase returns a set that always contains the parents of $X^j_t$ by Lemma S1 in \citet{runge2020discovering}. This still holds if latent context confounders are present, only additional links are possible. 

Now, we show soundness of the system-context links. 
If $X_t^j - C_{t-\tau}^i $ not in $ \hat{\mathcal{G}}^*$, then by Faithfulness also $X_t^j - C_{t-\tau}^i$ not in $ \mathcal{G}^*$.\\
For the other direction,if the link between $X_t^j$ and $C_{t-\tau}^i$ is not in $ \mathcal{G}^*$, due to the Causal Markov Condition it holds $(X_{t-\tau}^i, W^-_t) \indep C_t^j | \text{Pa}(C_t^j)$.
Define $W^-_t := (\hat{\mathcal{B}}^-_t(X_t^j) \setminus \{ C_{t-\tau}^i \}), \hat{\mathcal{B}}^-_{t-\tau}(C_{t-\tau}^i ) \setminus \text{Pa}(X^j_t)$ as in \cite{runge2020discovering} where $\hat{\mathcal{B}}^-_t(X_t)$ denotes the lagged adjacency set resulting from the lagged skeleton phase of PCMCI$^+$ (Algorithm 1 in \cite{runge2020discovering}).
Using the weak union property of conditional independence this implies $X_{t-\tau}^i \indep C_t^j | \text{Pa}(C_t^j), W^-_t$ which is, by definition of $W^-_t$ equivalent to $X_{t-\tau}^i \indep C_t^j | \text{Pa}(C_t^j), \hat{\mathcal{B}}_t^-(C_t^j) \setminus \{ X_{t-\tau}^i \}, \hat{\mathcal{B}}_{t-\tau}^-(X_{t-\tau}^i)$.
 Note that $\text{Pa}(C_t^j) \subset \mathbf{C} \cup \mathbf{L}$, however by Assumptions 2, there exist no latent confounders or mediators between system and context, thus we also find a set $S \subset \mathbf{C}$ s.t. $X_{t-\tau}^i \indep C_t^j | S, \hat{\mathcal{B}}_t^-(C_t^j) \setminus \{ X_{t-\tau}^i \}, \hat{\mathcal{B}}_{t-\tau}^-(X_{t-\tau}^i)$. This is tested at some iterative step of the algorithm and the link is removed.

Even though, eventually we are only interested in the soundness of system-context and system-system links, we need to establish that the dummy-system links within $\mathcal{G}_{alg}$ correspond to those in the dummy-projected ground truth graph $\mathcal{G}$. In the following, $D$ can either denote $D_\text{time}$ or $D_\text{space}$. \\
Now, we show if the link $D - X^j_t$ is not in $\mathcal{G}_{alg}$, then it also is not in $\mathcal{G}$. If the link $D - X^j_t$ is not in $\mathcal{G}_{alg}$, then $D \indep X_t^j | \mathbf{S}, \hat{\mathcal{B}}^C_t(X_t^j) $. This conditional independence also holds for all latent context nodes $L$ since it can be expressed as a non-invertible function of $D$. Therefore by Faithfulness and the non-invertibility of this function, there is also no link between $X^j_t$ and $L$ in the ground truth graph, and thus also no link to the dummy in its projected version $\mathcal{G}$.

For the other direction, let us define $W_t^- :=  \hat{\mathcal{B}}^-_t(X_t^j) \setminus \text{Pa}(X^j_t)$, the set $W_t^-$ does not contain parents of $X_t^j$, it also does not contain any latent nodes.
If the link $D - X^i_t$ is not in $\mathcal{G}$, then for all latent nodes $L$ it holds $L - X^i_t$ is not in the ground truth graph. 
Thus by the Causal Markov Condition $(L, W_t^-) \indep X_t^j | \text{Pa}(X^j_t)$, and by the weak union property and using the definition of $W_t^j$, we get $L \indep X_t^j | \text{Pa}(X^j_t), \hat{\mathcal{B}}^-_t(X_t^j)$ for all $L \in \mathbf{L}$. This also implies that $\text{Pa}_L(X^j_t) = \emptyset$.\\
Also, similarly to the non time series case, $X_t^j$ can be expressed as a function of the context nodes $\mathbf{C}$, $\mathbf{L}$ and the noise (and auto-correlation which is accounted for by $\hat{\mathcal{B}}^-_t(X_t^j)$).
This means, conditional on $(\text{Pa}(X_t^j) \setminus \mathbf{L}) \cup  \hat{\mathcal{B}}^-_t(X_t^j)$, $X_t^j$ only depends on the noise. The noise is independent of $D$, thus $X_t^j \indep D | \text{Pa}(X_t^j)$. Therefore, the algorithm will remove this dummy-system link.

Next, we show the soundness of the discovery of the system-system links.\\
We first show, if the link $X^i_{t-\tau} - X^j_t$ not in $ \hat{\mathcal{G}}^*$ then it is also not in $ \mathcal{G}^*$. Essentially, this follows with the same arguments as in non time series case combined with Faithfulness, but we will go through the arguments in more detail now.\\ To simplify the notation, we make the abbreviation $\mathcal{B} := \hat{\mathcal{B}}_{t}^-(X_{t}^j) \setminus \{ X^i_{t-\tau} \}, \hat{\mathcal{B}}_{t-\tau}^-(X_{t-\tau}^i)$. The algorithm removes the link between $X^i_{t-\tau}$ and $X^j_t$ if and only if
    \[
        X^i_{t-\tau} \indep X^j_t | S, \mathcal{B}, \text{Pa}_{CD}(X^i_{t-\tau}, X^j_{t})
    \]
for some $S \in \hat{\mathcal{A}}_t(X_t^j)$.\\
If $D_\text{time}, D_\text{space} \not \in \text{Pa}_{CD}(X^i_{t-\tau}, X^j_{t})$: Faithfulness is directly applicable.\\ 
Let now $D$ be either $D_\text{time}$ or $D_\text{space}$. If $D \in \text{Pa}_{CD}(X^i_{t-\tau}, X^j_{t})$ this yields, together with Lemma \ref{eq:dummy_cond_ts}
    \[
    X \indep Y | S, \mathcal{B}, \text{Pa}_{CD}(X, Y) ~ \implies ~ X \indep Y | S , \mathcal{B}, \{D\} ~ \implies ~ X \indep Y | S , \mathcal{B}, \mathbf{C}, \mathbf{L}.
    \]
and we can apply the Faithfulness argument.

 Now we show the other direction, i.e.\ if $X^i_{t-\tau} - X^j_t$ not in $\mathcal{G}^*  ~ \implies ~ $ $X^i_{t-\tau} - X^j_t$ not in $ \hat{\mathcal{G}}^* $. 
For that, we define $W^-_t := (\hat{\mathcal{B}}^-_t(X_t^j) \setminus \{ X^i_{t-\tau} \}, \hat{\mathcal{B}}^-_{t-\tau}( X^i_{t-\tau}), \text{Pa}_{C}(X^i_{t-\tau}) ) \setminus \text{Pa}(X^j_t)$ similar to \cite{runge2020discovering}. This set does not contain any parents of $ X^j_t$ and by the assumption also $X^i_{t-\tau}$ is not a parent of $ X^j_t$. Furthermore, we assume that for $\tau=0$, $ X^i_t$ is not a descendant of $ X^j_t$ (can be always achieved by exchanging the roles of $ X^i_t$ and $ X^j_t$).\\
Then the Causal Markov Condition implies $(X_{t-\tau}^i, W^-_t) \indep X_t^j | \text{Pa}(X_t^j)$
Using the weak union property this implies $X_{t-\tau}^i \indep X_t^j | \text{Pa}(X_t^j), W^-_t$ which is, by definition of $W^-_t$, equivalent to 
    \begin{equation} \label{eq:*}
        X_{t-\tau}^i \indep X_t^j | \text{Pa}(X_t^j), \hat{\mathcal{B}}_t^-(X_t^j) \setminus \{ X_{t-\tau}^i \}, \hat{\mathcal{B}}_{t-\tau}^-(X_{t-\tau}^i), \text{Pa}_{C}(X^i_{t-\tau})
    \end{equation}
    Note that the conditioning set potentially also contains nodes from $\mathbf{L}$ (but only in $\text{Pa}(X_t^j)$).
    This also implies 
    \begin{equation} \label{eq:**}
        X_{t-\tau}^i \indep X_t^j | \text{Pa}(X_t^j) \setminus \mathbf{L}, \hat{\mathcal{B}}_t^-(X_t^j) \setminus \{ X_{t-\tau}^i \}, \hat{\mathcal{B}}_{t-\tau}^-(X_{t-\tau}^i), \text{Pa}_C(X_{t-\tau}^j), \{ D_\text{time}, D_\text{space} \},
    \end{equation}
    and similarly if $\text{Pa}(X_t^j)$ only contains nodes from $\mathbf{L}_\text{space}$ which can be accounted for by additionally conditioning on $D_\text{space}$ (same argument holds if we replace space by time).
    If there are no latent nodes in $\text{Pa}(X_t^j)$, then (\ref{eq:*}) is the same as $X_{t-\tau}^i \indep X_t^j | \text{Pa}_{XC}(X_t^j), \hat{\mathcal{B}}_t^-(X_t^j) \setminus \{ X_{t-\tau}^i \}, \hat{\mathcal{B}}_{t-\tau}^-(X_{t-\tau}^i)$, in our algorithm we either test this or (\ref{eq:**}) and therefore remove the link.\\
    If there are latent nodes in $\text{Pa}(X_t^j)$, then $D \in \text{Pa}_{CD}(X_t^j)$ within $\mathcal{G}$ and thus also in $\text{Pa}_{CD}(X_t^j)$ within $\mathcal{G}_{alg}$ 
    , so we test (\ref{eq:**}) and remove the link.

For system-system links completeness follows as in \cite{runge2020discovering}. The context-system links are already correctly oriented by the exogeneity assumption (context cannot be a descendent of system).
\end{proof}

\begin{corollary}
If some of the observed context variables are treated as unobserved, and the assumptions 1, 2, 3, 4 still hold, our method J-PCMCI$^+$ will recover the correct system-system adjacencies.
\end{corollary}

\begin{proof}
    This follows directly from theorem \ref{thm:ts}.
\end{proof}

\section{Pseudocode}
We present the pseudocodes for \texttt{poolData, partialSkeletonPC} and \texttt{partialContempSkeletonPCMCI+} below. 
\RestyleAlgo{ruled}
\begin{algorithm}
\caption{poolData (for non-time-series data)\\
For the time-series case we rely on the functionality supplied in Tigramite \citep{runge2019detecting} to handle time-series data from multiple data-sets while keeping the time structure (in particular while using the sliding window approach to cunstruct time-series data for the lagged variables $X^i_{-\tau}$ for $\tau> 0$}
\KwData{$M$ data-sets $\mathbf{X}$ containing observations of the same system (and for time-series case: temporal context) variables, $M$ observations of context variables $\mathbf{C}$ (one per data-set), optional: dummy variable with $M$ distinct values}
\KwResult{one data-set containing the pooled data}
Let $N$ denote the number of system variables\\
Let $K$ denote the number of (observed) context variables\\
Let $T_m$ denote the sample-size of the system variables with $m=1, \ldots, M$\\
\For{$i$ in $1, \ldots, N$}{
concatenate $(X^{i, (m)})_{m=1, \ldots, M}$
}
\For{$j$ in $1, \ldots, K$}{
construct array of context variable $C^j$ by repeating its $M$ values $T_m$ times
}
\If{data for the dummy variable $D$ is provided}{
construct array for the dummy variable $D$ by repeating its $M$ values $T_m$ times\\
}

return $(X^1, \ldots, X^N, C^1, \ldots, C^K, D)$
\end{algorithm}

\RestyleAlgo{ruled}
\begin{algorithm}
\caption{partialSkeletonPC\\
$\operatorname{CI}(X, Y, S)$ is some suitable conditional independence test}
\KwData{Data $\mathbf{X}$, significance level $\alpha$, node pairs to consider $\mathcal{P}$, link knowledge $\mathcal{B}$}
\KwResult{graph $\mathcal{G}$}
Form a graph $\mathcal{G}$ with information from $\mathcal{B}$, connect all other nodes with undirected links \\
Set $p=0$\\
\While{any adjacent pairs $(X, Y)$ in $\mathcal{P}$ satisfy $|\mathcal{A}(X) \setminus \{Y\}| \geq p$}{
    Select an adjacent pair $(X, Y)$ from $\mathcal{P}$ with $|\mathcal{A}(X) \setminus \{Y\}| \geq p$\\
    Select $S \subset \mathcal{A}(X) \setminus \{Y\}$ with $|S| = p$\\
    p-value $\leftarrow \operatorname{CI}(X, Y, S)$\\
    \If{p-value $>\alpha$}{
         Delete link $X - Y$ from $\mathcal{G}$\\
         Store (unordered) sepset $(X, Y) = \mathcal{S}$
         }
}
return $\mathcal{G}$, sepset
\end{algorithm}

\RestyleAlgo{ruled}
\begin{algorithm*}
\caption{partialContempSkeletonPCMCI+, small adaption of Algorithm 2 in \cite{runge2020discovering}\\
$\operatorname{CI}(X, Y, S)$ is some suitable conditional independence test}
\KwData{$M$ time-series data-sets $X^{(m)}=(X^{1,(m)}, \ldots, X^{N,(m)})$ which can contain system, context and also dummy variables, indices of system variables $J$, max. time lag $\tau_\text{max}$, significance threshold $\alpha_\text{PC}$, $\hat{\mathcal{B}}^-_{t}(X_t^j)$ for all $X_t^j \in \mathbf{X}_t = (X_t^{1}, \ldots, X_t^{N})$\footnote{Note that we drop the data-set index $(m)$ whenever we refer to the nodes (and not the variables).}, contextual parents $\mathcal{C}(X)$, pairs to consider $\mathcal{P}$}
\KwResult{graph $\mathcal{G}$, sepset}
 Form time series graph $\mathcal{G}$ with lagged links from $\hat{\mathcal{B}}^-_{t}(X_t^j)$ for all $X_t^j \in \mathbf{X}_t$, fully connect all contemporaneous system variables, i.e.\, add $X^i_t - X^j_t$ for all $X_t^i \neq X_t^j \in \mathbf{X}_t $ with $i,j \in J$,  and set links between context and system according to $\mathcal{C}(X)$  \\
 Initialize contemporaneous adjacencies $\hat{\mathcal{A}}(X_t^j) := \hat{\mathcal{A}}_t(X_t^j) = \{ X_t^i \neq X_t^j \in \mathbf{X}_t | X^i_t - X^j_t \text{ in } \mathcal{G} \}$\\
 Let $p = 0$\\
 \While{any adjacent pairs $(X^i_{t-\tau}, X^j_t)$ for $\tau \geq 0$ in $\mathcal{G}$ from $\mathcal{P}$ satisfy $| \hat{\mathcal{A}}(X_t^j) \setminus \{ X_{t-\tau}^j \} | \geq p$}{
 Select new adjacent pair $(X^i_{t-\tau}, X^j_t)$ from $\mathcal{P}$ for $\tau \geq 0$ satisfying $| \hat{\mathcal{A}}(X_t^j) \setminus \{ X_{t-\tau}^j \} | \geq p$ \\
     \While{$(X^i_{t-\tau}, X^j_t)$ are adjacent in $\mathcal{G}$ and not all $S \subset \hat{\mathcal{A}}(X_t^j) \setminus \{ X_{t-\tau}^j \}$ with $|S| = p$ have been considered }{
     Choose new $S \subset \hat{\mathcal{A}}(X_t^j) \setminus \{ X_{t-\tau}^j \}$ with $|S| = p$\\
     Set $\mathbf{Z}=(S, \hat{\mathcal{B}}^-_{t}(X_t^j) \setminus \{ X_{t-\tau}^j \}, \hat{\mathcal{B}}^-_{t-\tau}(X_{t-\tau}^j)$\\
     $(X^i_{t-\tau}, X^j_t, \mathbf{Z}) \leftarrow \operatorname{poolData}((X^{i, (m)}_{t-\tau}, X^{j, (m)}_t, \mathbf{Z}^{(m)})_{m=1, \ldots, M})$\\
     $(\text{p-value}, I) \leftarrow \text{CI}(X^i_{t-\tau}, X^j_t, \mathbf{Z})$\\
     $I^\text{min}(X^i_{t-\tau}, X^j_t) = \min(|I|, I^\text{min}(X^i_{t-\tau}, X^j_t))$\\
         \If{p-value $>\alpha_\text{PC}$}{
         Delete link $X^i_{t-\tau} \rightarrow X^j_t$ for $\tau > 0$ (or $X^i_t - X^j_t$ for $\tau = 0$) from $\mathcal{G}$\\
         Store (unordered) sepset $(X^i_{t-\tau}, X^j_t) = \mathcal{S}$
         }
     Let $p=p+1$ and re-compute $\hat{\mathcal{A}}(X_t^j)$ from $\mathcal{G}$ and sort by $I^\text{min}(X^i_{t-\tau}, X^j_t)$ from largest to smallest\\
     }
 }
return $\mathcal{G}$, sepset
\end{algorithm*}

\newpage
\section{Simplified experimental setup}\label{sec:simple_setup}

We want to understand the shape of the adjacency-FPR surface of our method better. From the numerical results of the standard setup it seems that for a fixed samplesize $T$ the FPR goes up with the number of datasets. On the other hand, for a fixed number of datasets $M$, FPR goes down with increasing samplesize. A similar pattern is visible when simply applying PCMCI$^+$ on data where dummy variables have been included. 

To this end, we simplify our experimental setup in the following way. We sample data from a specific version of the SCM (1):
    \begin{equation}
        \begin{split}
            X^0_t &:=  0.5X^1_{t}+0.5C_\text{space}^0 + 0.5C_\text{space}^1 +0.5C_{\text{time}, t-1}^0 + 0.5C_{\text{time}, t-1}^1+ \eta^0\\
            X^1_t &:= 0.5X^1_{t-1}+ 0.5C_\text{space}^0 + 0.5C_\text{space}^1 +0.5C_{\text{time}, t-1}^0 +0.5C_{\text{time}, t-1}^1 + \eta^1\\
            C_\text{space}^0 &:= \eta^0_\text{space}\\
            C_\text{space}^1 &:= \eta^1_\text{space}\\
            C_{\text{time},t}^0 &:= \eta^0_\text{time}\\
            C_{\text{time},t}^1 &:= \eta^1_\text{time},
        \end{split}
    \end{equation}
where $C_\text{space}^1$ and $C_\text{time}^1$ are unobserved, and all other variables are observed.
On the system data of this SCM we apply a modified version of PCMCI$^+$ where we always including the dummy variables in the conditioning sets of all conditional independence tests. By doing so, we are able to see what effect conditioning on the dummies has on the FPR, see figure \ref{fig:simple_setup}. In the FPR-plot, we see a similar pattern as is visible in the more involved experimental setup. Generally speaking, for a fixed samplesize $T$ the FPR goes up with the number of datasets $M$ while, for a fixed number of datasets $M$, it goes down with increasing samplesize $T$. This means, in the large sample limit (of both $M$ and $T$) we can expect consistent results. However, if we only have a small samplesize $T$, there are potentially inflated false positives.

\begin{figure}[htb!]
    \centering
    \includegraphics[width=\linewidth]{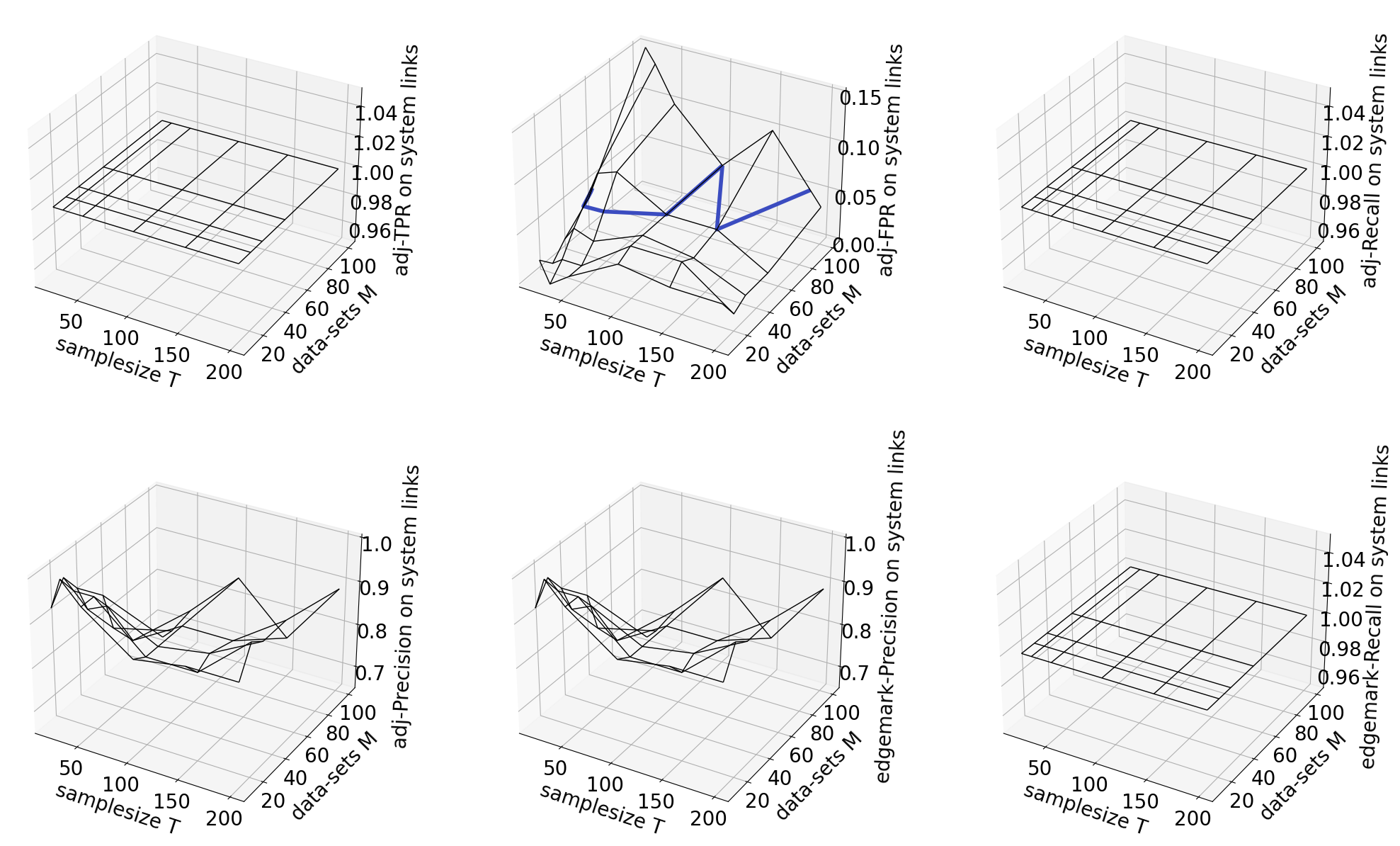}
    \caption{Discovery results of system-system links in the simplified experimental setup (section \ref{sec:simple_setup}) for varying sample sizes $T$, and number of datasets $M$.}
    \label{fig:simple_setup}
\end{figure}

\section{Nonlinear experimental setup}\label{sec:simple_setup_nonlin}

We extend the simplified experimental setup of section \ref{sec:simple_setup} a bit to allow for nonlinear mechanisms. In this way, we are able to demonstrate that our method can be flexibly combined with any CI test. In this setup, we use a CI test based on Gaussian process regression and a distance correlation (GPDC).

\begin{equation}
    \begin{split}
        &X^0_t :=  0.3(X^1_{t})^2+0.5C_\text{space}^0 - 0.2(C_{\text{time}, t-1}^0)^2 + \eta^0\\
        &X^1_t := 0.5X^1_{t-1}- 0.5(C_\text{space}^0)^2 + 0.3(C_{\text{time}, t-1}^0)^2 + \eta^1\\
        &C_\text{space}^0 := \eta^0_\text{space}\\
        &C_{\text{time},t}^0 := \eta^0_\text{time}
    \end{split}
\end{equation}

We show the results in figure \ref{fig:simple_setup_nonlin}.

\begin{figure}
    \centering
    \includegraphics[width=0.6\linewidth]{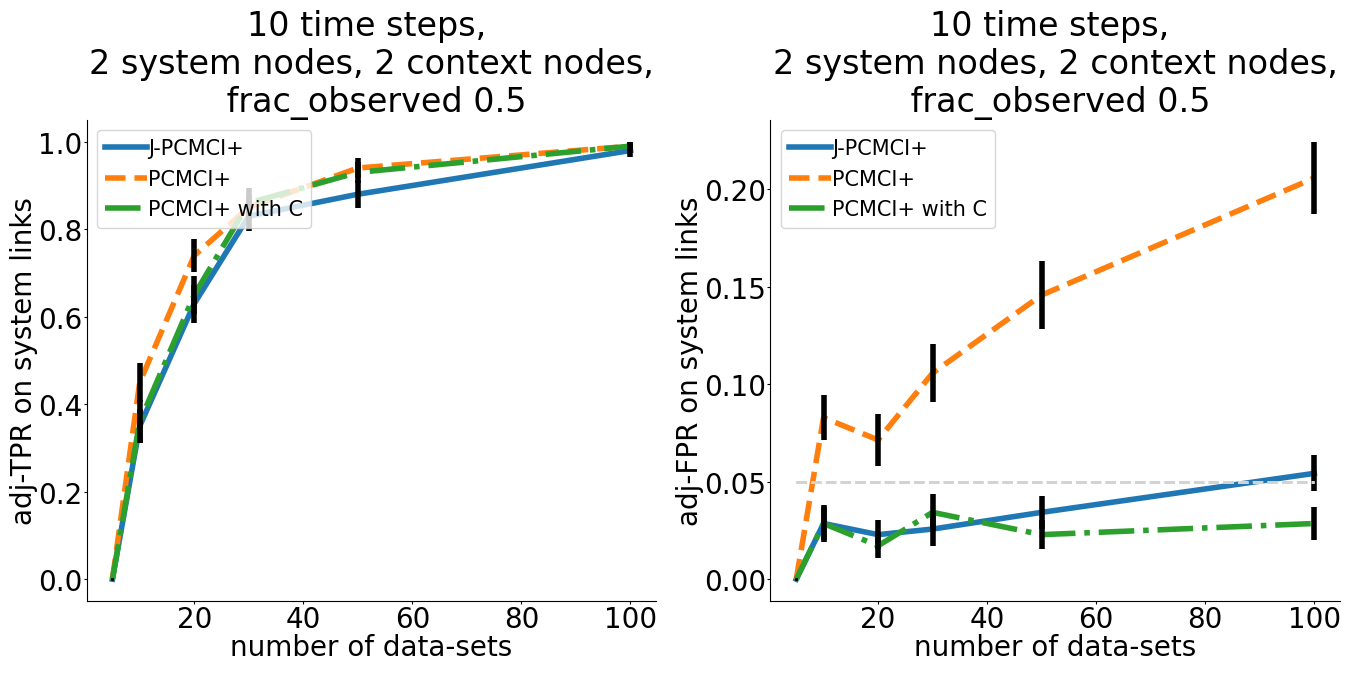}
    \includegraphics[width=0.6\linewidth]{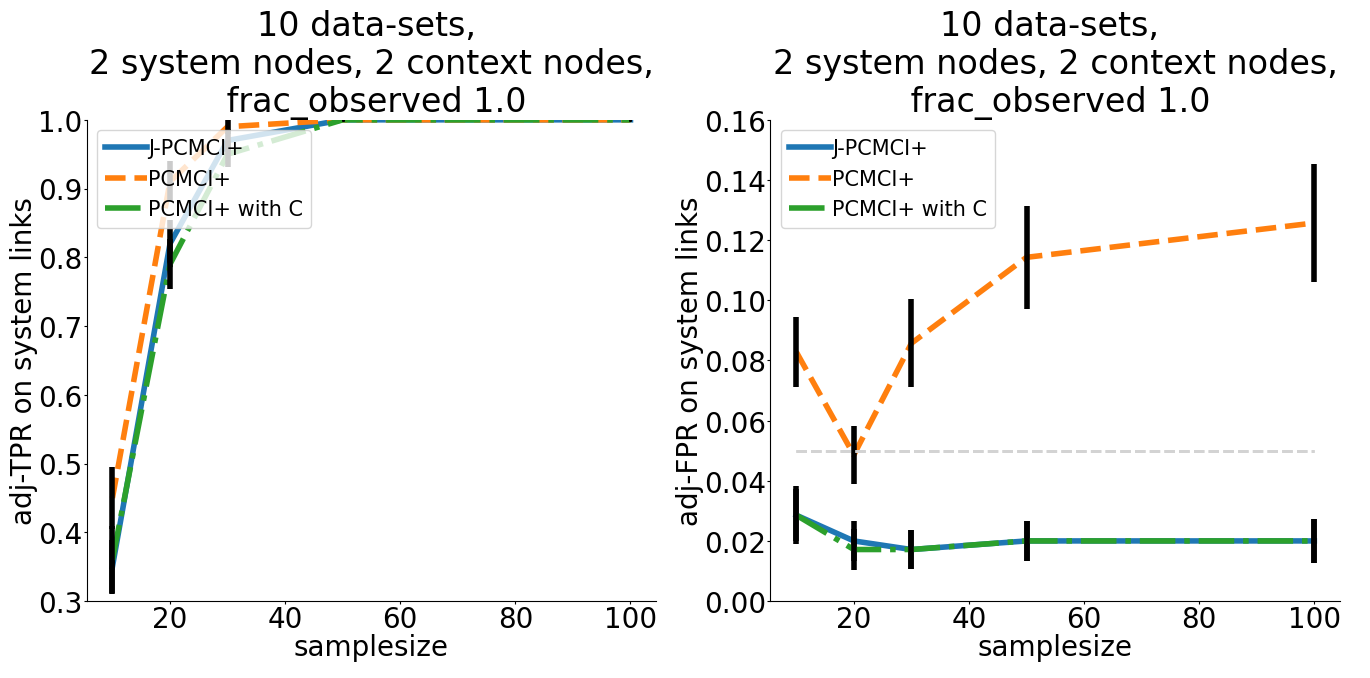}
    \caption{Discovery results of system-system links for varying sample sizes $T$, and fixed $M=10$ (top row), and varying number of contexts $M$, and fixed $T=10$ (bottom row). The data is generated according to the SCM described in section \ref{sec:simple_setup_nonlin}. In this setting all of the context nodes are observed. We compare our method (J-PCMCI+) to PCMCI$^+$ using all data of observed nodes (PCMCI+ with C) and only using data of system variables (PCMCI+).}
    \label{fig:simple_setup_nonlin}
\end{figure}

\section{Additional Plots}
\begin{figure}
    \centering
    \includegraphics[width=0.8\linewidth]{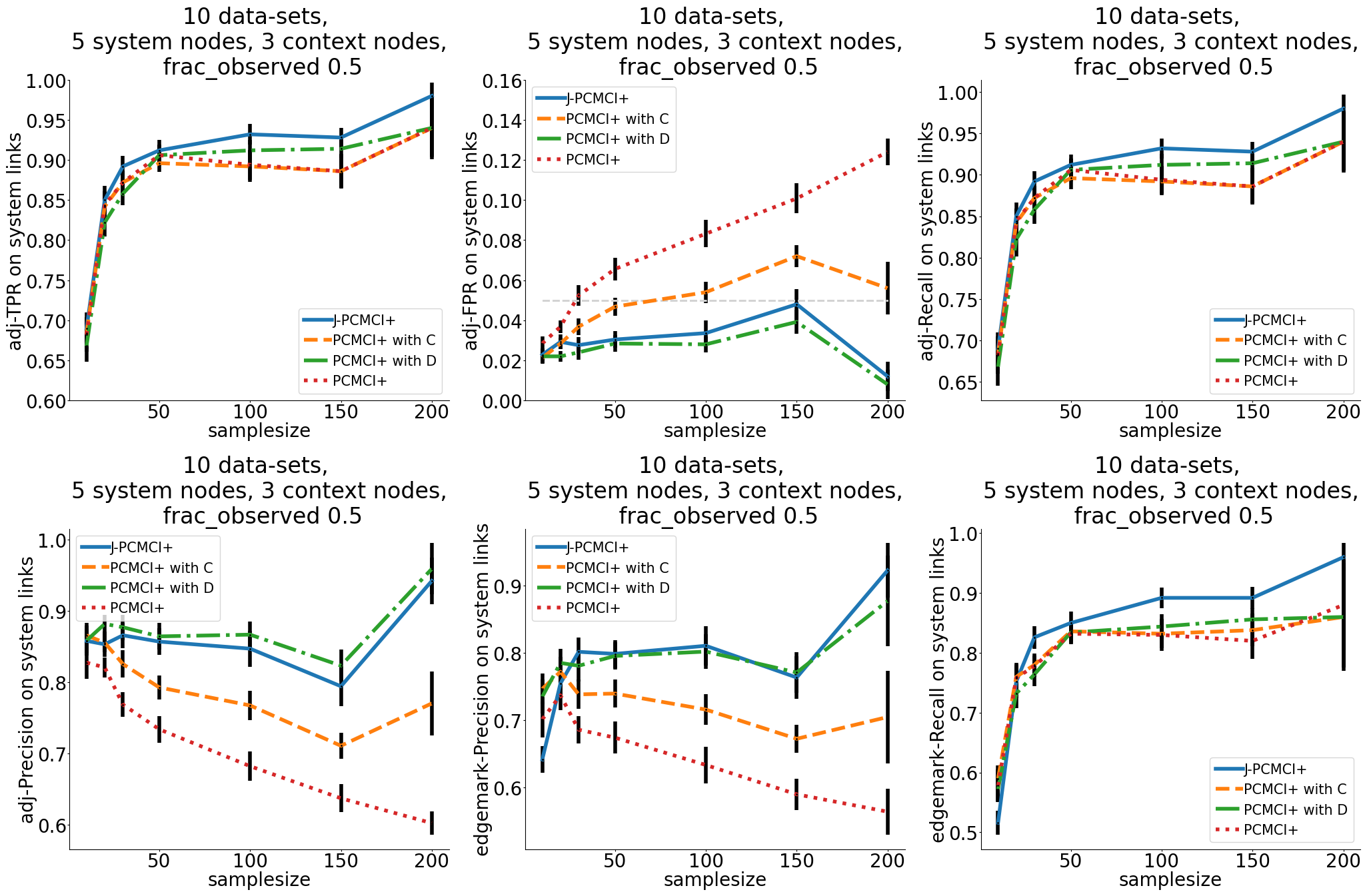}\\
    ~\\
    \includegraphics[width=0.8\linewidth]{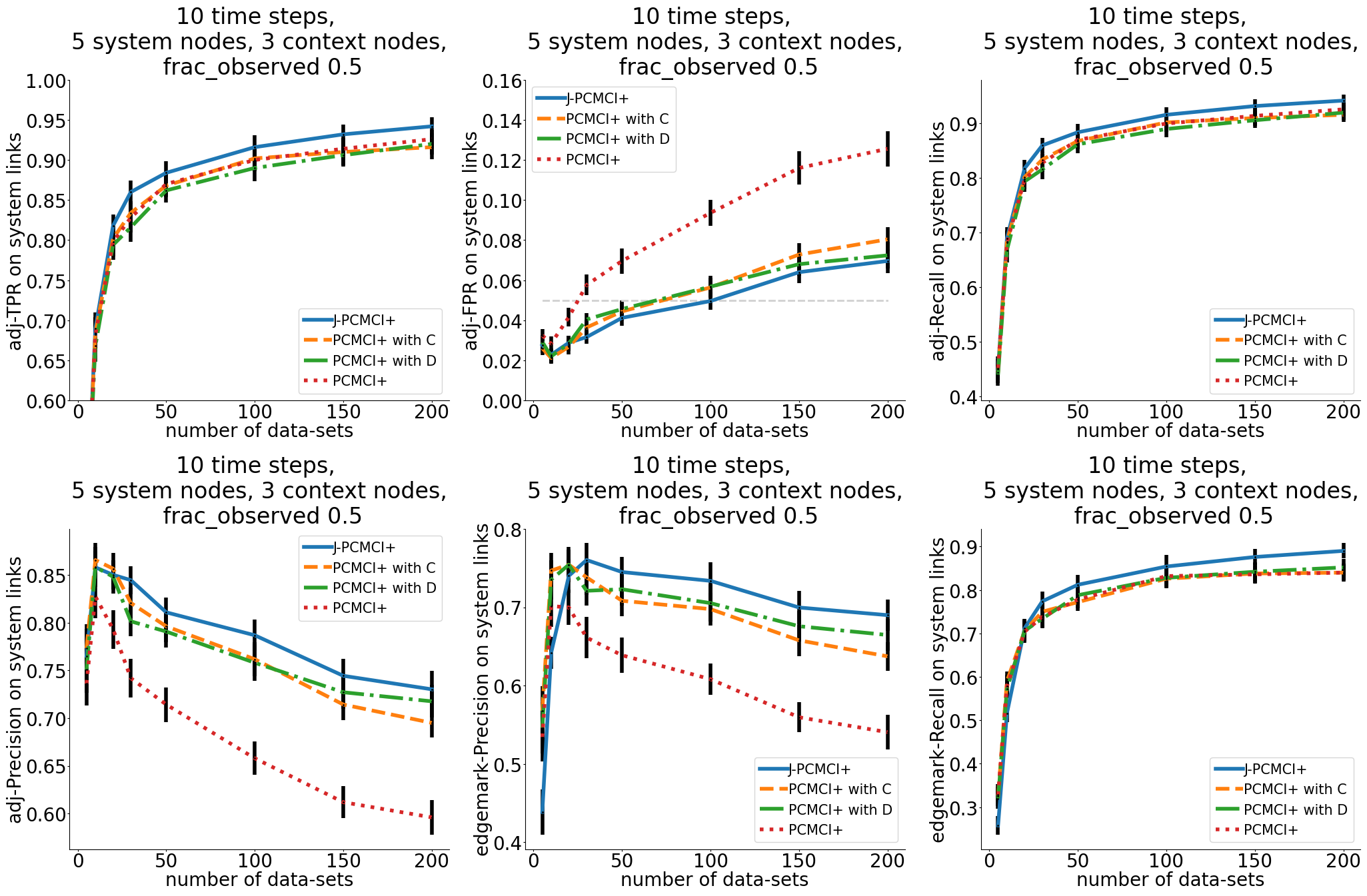}
    \caption{Discovery results of system-system links for varying sample sizes $T$, and fixed $M=10$ (top two rows), and varying number of contexts $M$, and fixed $T=10$ (bottom two rows). All other setup parameters are set as the defaults described in the main text. In this setting half of the context nodes are observed. We compare our method (J-PCMCI+) to PCMCI$^+$ using all data of observed nodes (PCMCI+ with C), using all data of system variables and including dummies (PCMCI+ with D), and only using data of system variables (PCMCI+).}
\end{figure}

\begin{figure}
    \centering
    \includegraphics[width=0.8\linewidth]{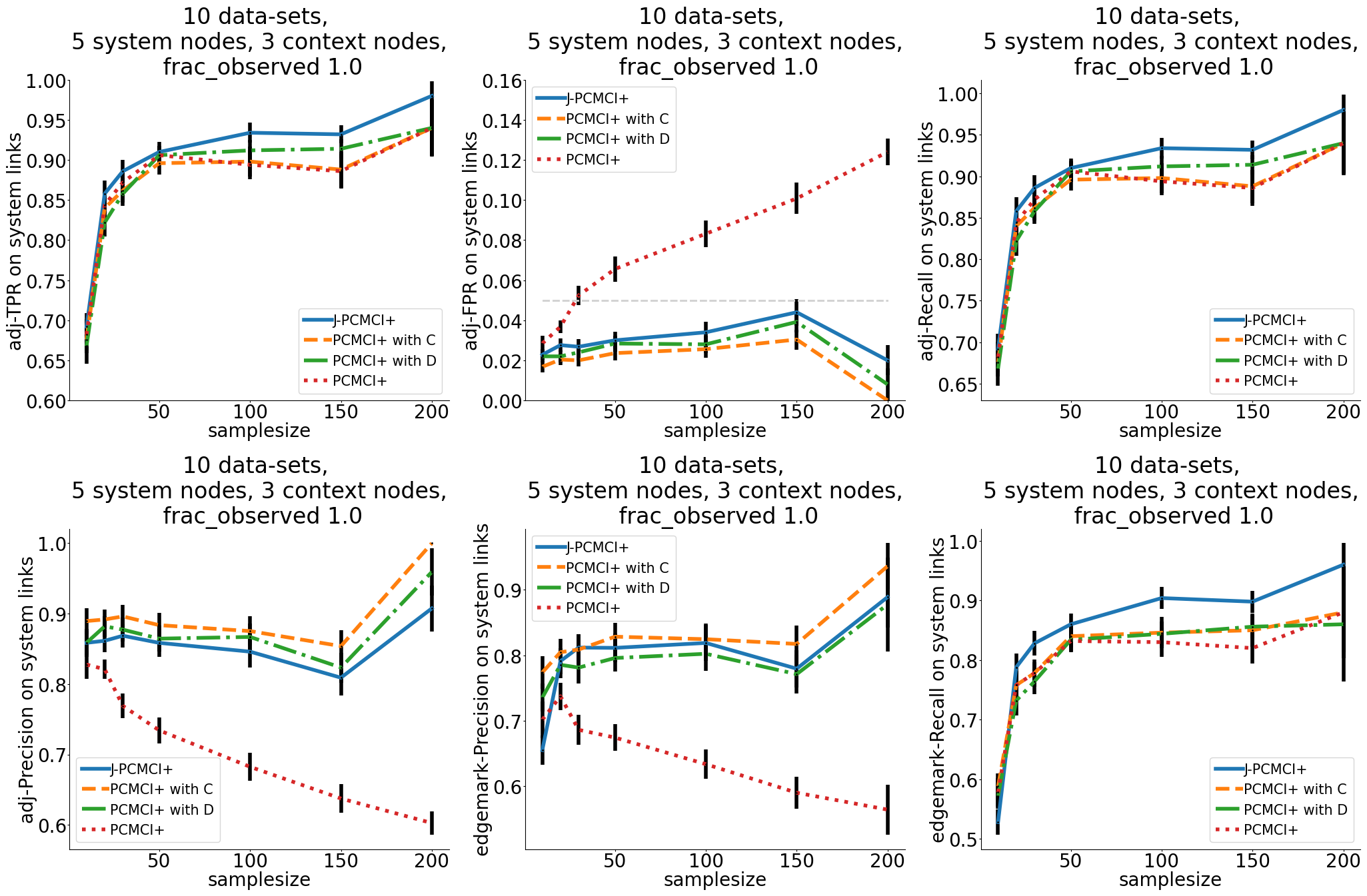}
    \includegraphics[width=0.8\linewidth]{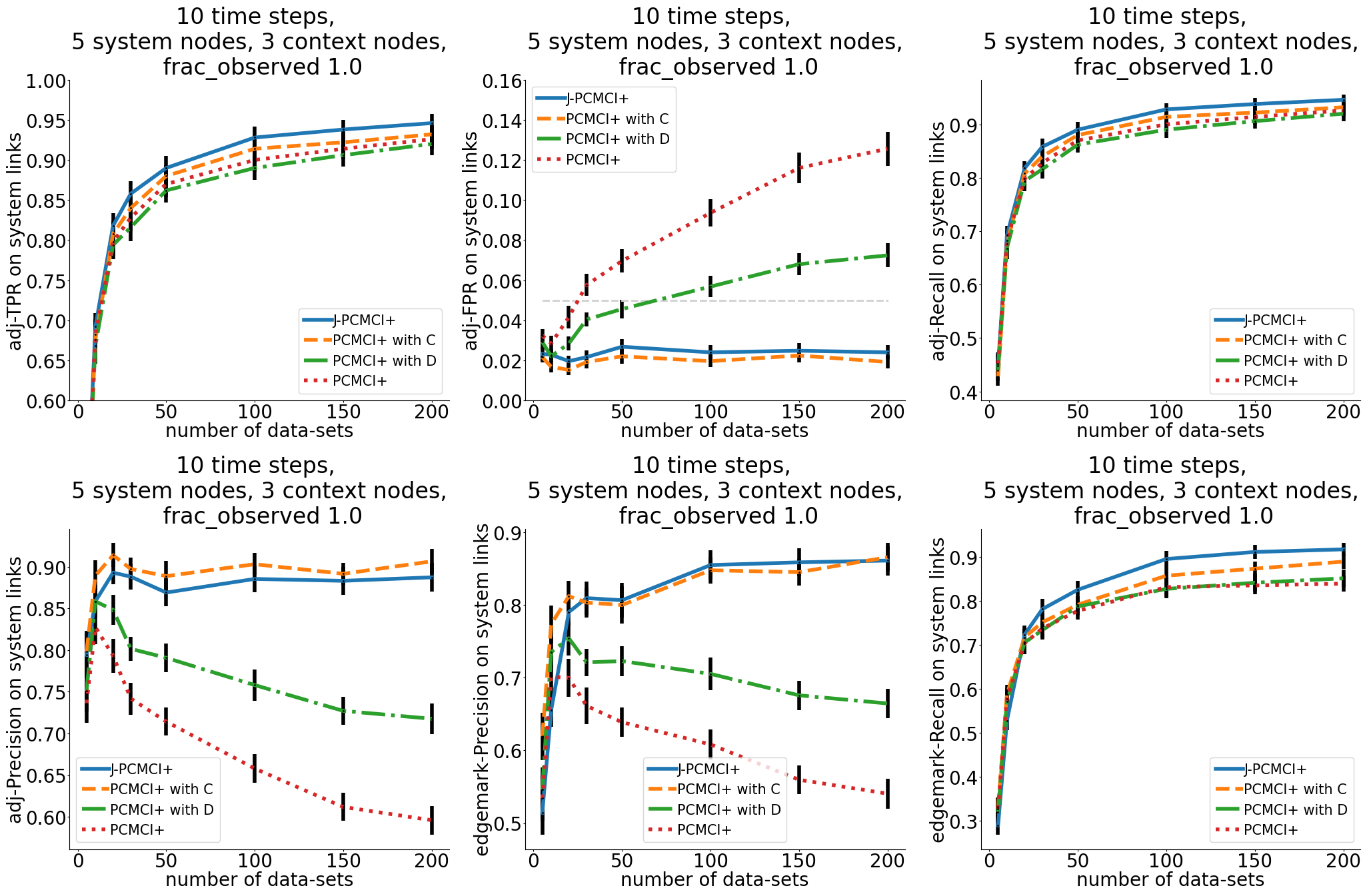}
    \caption{Discovery results of system-system links for varying sample sizes $T$, and fixed $M=10$ (top two rows), and varying number of contexts $M$, and fixed $T=10$ (bottom two rows). All other setup parameters are set as the defaults described in the main text. In this setting all of the context nodes are observed. We compare our method (J-PCMCI+) to PCMCI$^+$ using all data of observed nodes (PCMCI+ with C), using all data of system variables and including dummies (PCMCI+ with D), and only using data of system variables (PCMCI+).}
\end{figure}

\begin{figure}
    \centering
    \includegraphics[width=\linewidth]{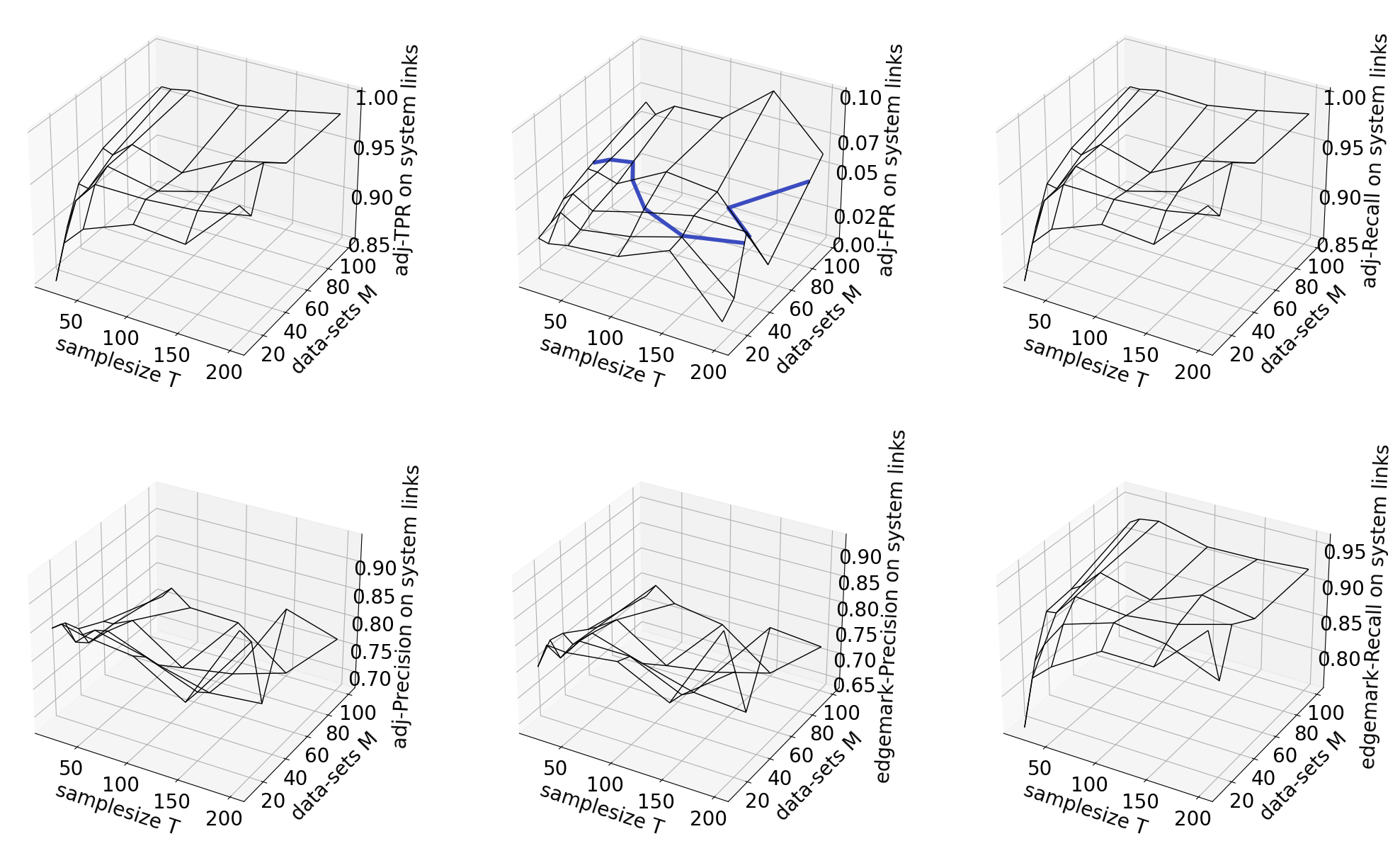}
    \caption{Discovery results of our method (J-PCMCI+) on system-system links for varying sample sizes $T$, and number of contexts $M$. All other setup parameters are set as the defaults described in the main text. In this setting half of the context nodes are observed. We show the contour line corresponding to the significance level $\alpha$ in the adjacency-FPR plot.}
\end{figure}

\begin{figure}
    \centering
    \includegraphics[width=\linewidth]{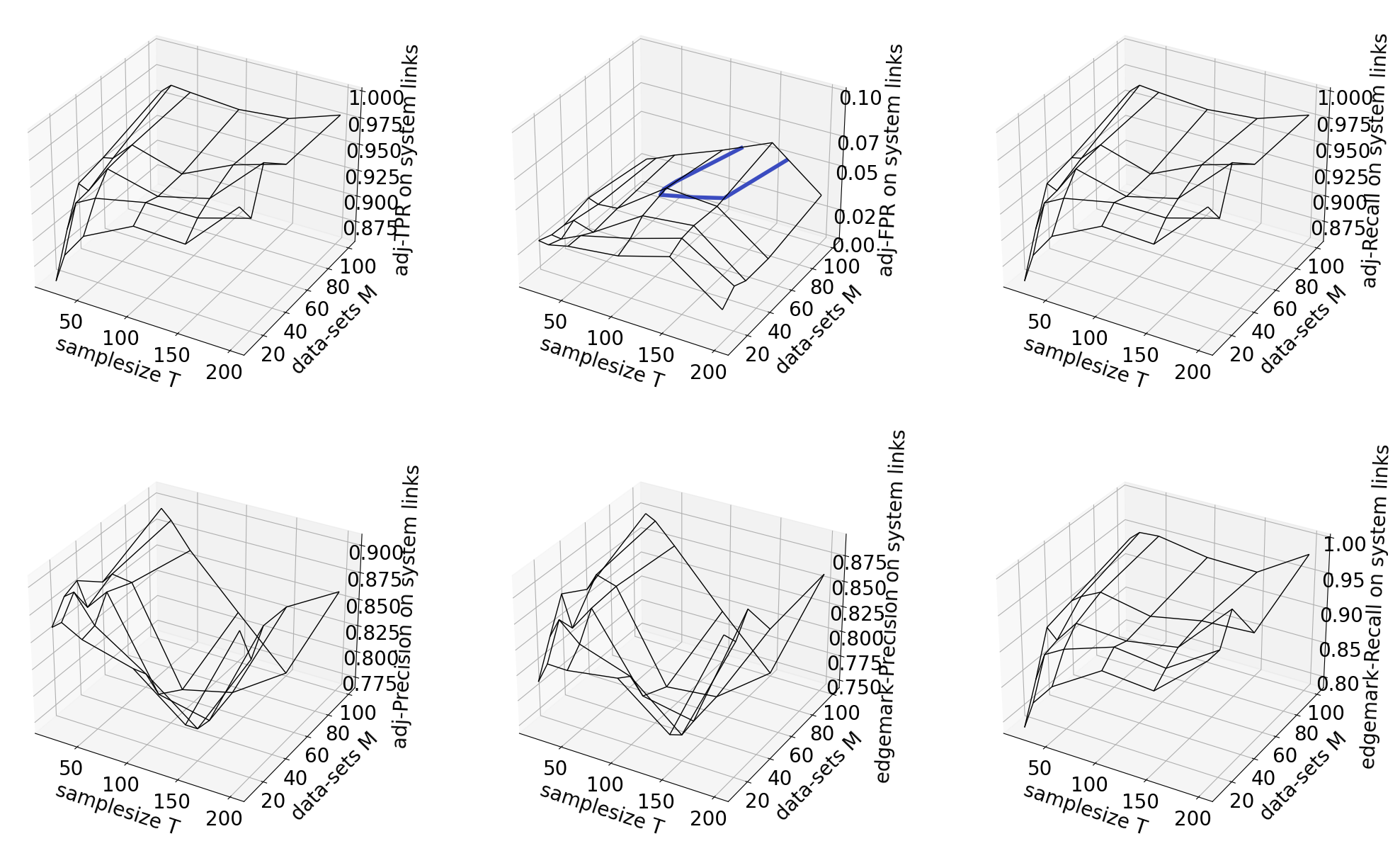}
    \caption{Discovery results of our method (J-PCMCI+) on system-system links for varying sample sizes $T$, and number of contexts $M$. All other setup parameters are set as the defaults described in the main text. In this setting all the context nodes are observed.We show the contour line corresponding to the significance level $\alpha$ in the adjacency-FPR plot.}
\end{figure}

\begin{figure}
    \centering
    \includegraphics[width=0.7\linewidth]{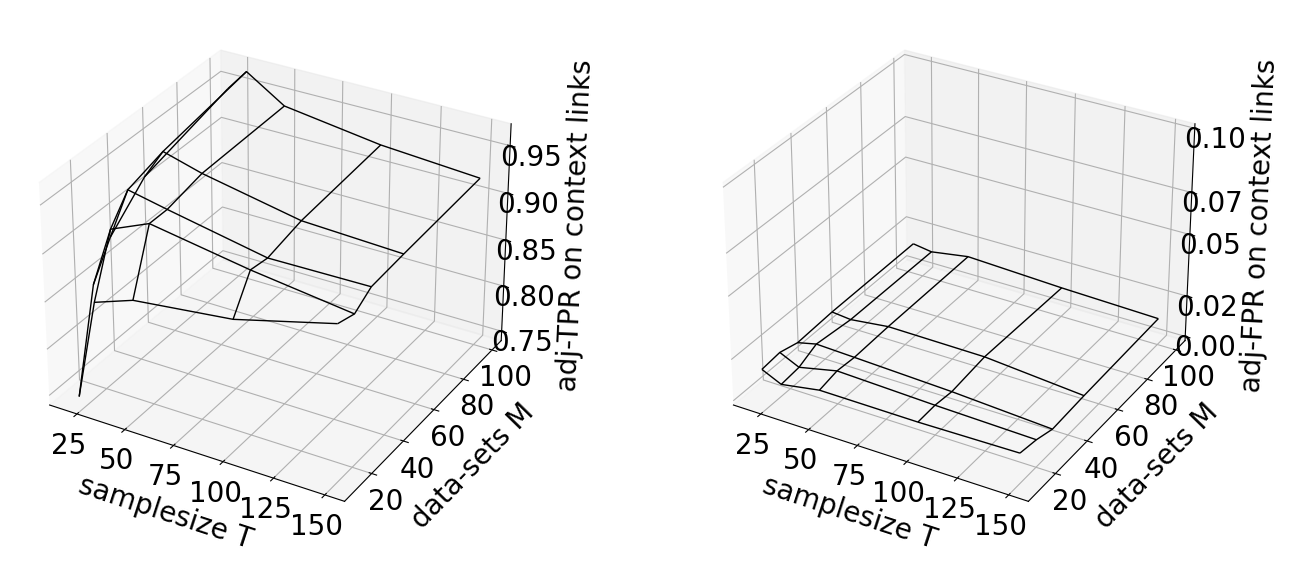}
    \caption{Discovery results of our method (J-PCMCI+) on context-system links for varying sample sizes $T$, and number of contexts $M$. All other setup parameters are set as the defaults described in the main text. In this setting all the context nodes are observed.}
\end{figure}
\newpage
\bibliography{gunther_390}